\documentclass[letterpaper]{amsart}
\usepackage{latexsym,amssymb,enumerate,mathrsfs,amsthm,graphicx}
\usepackage{pb-diagram}

\usepackage[small]{caption}

\newcommand{\alg}[1]{\mathfrak{#1}}
\newcommand{\dhr}{\Delta_{\textrm{DHR}}}
\newcommand{\bfm}{\Delta_{\textrm{BF}}}
\newcommand{\id}{\operatorname{id}}
\newcommand{\Ad}{\operatorname{Ad}}
\newcommand{\Hom}{\operatorname{Hom}}
\newcommand{\End}{\operatorname{End}}

\newcommand{\lspan}{\operatorname{span}}
\newcommand{\spc}[1]{\mathscr{#1}}
\newcommand{\mc}[1]{\mathcal{#1}}

\providecommand{\norm}[1]{\ensuremath{\lVert#1\rVert}}



\theoremstyle{plain}
\newtheorem{theorem}{Theorem}
\numberwithin{theorem}{section}
\newtheorem{lemma}[theorem]{Lemma}
\newtheorem{prop}[theorem]{Proposition}
\newtheorem{corollary}[theorem]{Corollary}
\newtheorem{remark}[theorem]{Remark}

\numberwithin{equation}{section}

\newtheorem{definition}[theorem]{Definition}
\newtheorem*{property}{Property B}

\title[Extension of stringlike localised sectors in $d=2+1$]{On the extension of stringlike localised sectors in 2+1 dimensions}
\author{Pieter Naaijkens}
\email{p.naaijkens@math.ru.nl}
\date{\today}

\address{Radboud University Nijmegen \\
IMAPP \\
FNWI \\
Heyendaalseweg 135 \\
6525 AJ Nijmegen \\
The Netherlands}

\begin{document}
\begin{abstract}
	In the framework of algebraic quantum field theory, we study the category $\Delta_{\textrm{BF}}^{\mathfrak{A}}$ of stringlike localised representations of a net of observables $\mathcal{O} \mapsto \mathfrak{A}(\mathcal{O})$ in three dimensions. It is shown that compactly localised (DHR) representations give rise to a non-trivial centre of $\Delta_{\textrm{BF}}^{\mathfrak{A}}$ with respect to the braiding. This implies that $\Delta_{\textrm{BF}}^{\mathfrak{A}}$ cannot be modular when non-trivial DHR sectors exist. Modular tensor categories, however, are important for topological quantum computing. For this reason, we discuss a method to remove this obstruction to modularity.

	Indeed, the obstruction can be removed by passing from the observable net $\mathfrak{A}(\mathcal{O})$ to the Doplicher-Roberts field net $\mathfrak{F}(\mathcal{O})$. It is then shown that sectors of $\mathfrak{A}$ can be extended to sectors of the field net that commute with the action of the corresponding symmetry group. Moreover, all such sectors are extensions of sectors of $\mathfrak{A}$. Finally, the category $\Delta_{\textrm{BF}}^{\mathfrak{F}}$ of sectors of $\mathfrak{F}$ is studied by investigating the relation with the categorical crossed product of $\Delta_{\textrm{BF}}^{\mathfrak{A}}$ by the subcategory of DHR representations. Under appropriate conditions, this completely determines the category $\Delta_{\textrm{BF}}^{\mathfrak{F}}$.
\end{abstract}

\maketitle
\thispagestyle{empty}
\section{Introduction}
The study of superselection sectors and particle statistics has been a long-standing subject in algebraic quantum field theory~\cite{MR1405610}. Superselection sectors can be described as representations of a local net $\mathcal{O} \mapsto \alg{A}(\mc{O})$ of observables. The physically relevant representations are selected by a certain selection criterion. A superselection sector, then, is a (unitary) equivalence class of representations satisfying this criterion. These representations can be shown to have the structure of a tensor category resembling the category of representations of a compact group. In this category, one can define a braiding, closely related to the statistics of sectors.

It is well known that for the compactly localised representations first considered by Doplicher, Haag and Roberts, the braiding is in fact symmetric in spacetimes of dimension three or higher~\cite{MR1016869}. However, if one considers the weaker condition of localisation in some ``fattening string'' extending to spacelike infinity, the braiding is non-symmetric for spacetimes of dimension 3 or less~\cite{MR1104414}. Buchholz and Fredenhagen have shown that for massive particle states, this localisation condition holds~\cite{MR660538}.

The category of such stringlike localised representations in three dimensions automatically satisfies most of the axioms of a modular tensor category~\cite{MR1797619,MR1292673}. This class of tensor categories plays a prominent role in the theory of \emph{topological quantum computation}, see e.g.~\cite{MR1943131,MR1910833,MR1951039,MR2200691,PanangadenPaquette}. A good review can be found in~\cite{MR2443722}. This is one of the reasons why modular tensor categories are interesting, providing a reason to investigate if we can obtain modular tensor categories from algebraic quantum field theory. Another part of the motivation is provided by related constructions and results in e.g.~\cite{MR1838752,MR2183964,MR1128146}, where the extension of compactly localised representations in $d=1+1$ is discussed.

First, we give a brief overview of the basics of algebraic quantum field theory (AQFT), also called \emph{local quantum physics}. The leading idea in AQFT is that local algebras of observables encode all relevant information of a given physical theory. For each double cone $\mathcal{O}$ in Minkowski space $\mathbb{M}^3$ there is an associated unital $C^*$-algebra $\alg{A}(\mc{O})$ of observables, which are said to be localised in $\mc{O}$. This assignment of observable algebras should satisfy the following properties:
\begin{enumerate}
	\item \emph{Isotony:} if $\mc{O}_2 \subset \mc{O}_2$ then $\alg{A}(\mc{O}_1) \subset \alg{A}(\mc{O}_2)$. We assume the inclusions are injective unital $*$-homomorphisms.
	\item \emph{Locality:} if $\mc{O}_1$ is spacelike separated from $\mc{O}_2$, then the associated local observable algebras commute.
	\item \emph{Translation covariance:} there is a strongly continuous action $x \mapsto \beta_x$ of the translation group $\mathbb{M}^3$ on the local algebras, such that $\beta_x(\alg{A}(\mc{O})) = \alg{A}(\mc{O}+x)$.
\end{enumerate}
To avoid the trivial case we assume in addition that for each double cone $\mc{O}$ the algebra $\alg{A}(\mc{O})$ contains an element that is not a multiple of the identity. Note that the set of double cones in $\mathbb{M}^3$ is directed by inclusion. The inductive limit of this net in the category of $C^*$-algebras is denoted by $\alg{A}$ and is called the \emph{quasi-local algebra}. By means of a specific faithful irreducible representation $\pi_0: \alg{A} \to \alg{B}(\mc{H}_0)$, typically the vacuum representation, $\alg{A}$ is represented as a net of bounded operators on a Hilbert space $\mc{H}_0$. It is then natural to consider $\pi_0(\alg{A}(\mc{O}))''$ for each $\mc{O}$, where the prime denotes the commutant. This leads to net of \emph{von Neumann algebras}, which we will again denote by $\alg{A}(\mc{O})$. This net turns out to be more convenient to work with, and thus we will from now on assume that $\alg{A}(\mc{O})$ is a von Neumann algebra for each $\mc{O}$. The algebra $\alg{A}$ again will be the \emph{norm} closure of the union of these local (von Neumann) algebras. Note that $\alg{A}$ is not a von Neumann algebra in general.

The vacuum representation $\pi_0$ must satisfy a few additional conditions. It should be covariant under translations, say with a strongly continuous group of unitaries $\mc{U}_0(x)$, $x \in \mathbb{M}^3$. There is a unique (up to a phase) vacuum vector $\Omega$ such that $\mc{U}_0(x) \Omega = \Omega$ for all $x$. Moreover, the \emph{spectrum condition} for the generators of translations should hold: the joint spectrum of the generators of the translations should be contained in the forward lightcone $\overline{V^{+}}$. For details and motivations see e.g.~\cite{MR1768634}. Buchholz and Fredenhagen provide a construction that, given a massive single particle representation, produces a corresponding vacuum representation $\pi_0$ satisfying these criteria~\cite{MR660538}.

A superselection sector is then a unitary equivalence class of representations of $\alg{A}$ satisfying a certain (physically motivated) selection criterion. For example, Buchholz and Fredenhagen were led to consider stringlike localised sectors~\cite{MR660538}. The category of these representations, denoted by $\bfm^{\alg{A}}$, has a very rich structure. An essential ingredient in the analysis of this structure is the axiom of \emph{Haag duality}, which strengthens locality. If $\spc{S}$ is some unbounded region of spacetime, the $C^*$-algebra $\alg{A}(\spc{S})$ is defined by
\[
\alg{A}(\spc{S}) = \overline{\bigcup_{\mc{O} \subset \spc{S}} \alg{A}(\mc{O})}^{\|\cdot\|},
\]
where the closure in norm is taken and the union is taken over all double cones contained in $\spc{S}$. Suppose $\spc{S}$ is any connected causally complete region, that is, $\spc{S} = (\spc{S}')'$, where the prime denotes taking the causal complement. Haag duality then is the condition that
\begin{equation}
  \label{eq:haagd}
  \pi_0(\alg{A}(\spc{S}'))' = \pi_0(\alg{A}(\spc{S}))''.
\end{equation}
Here the prime in $\spc{S}'$ denotes taking the causal complement, whereas the other primes stand for the commutant. We will only need this duality relation in the case where $\spc{S}$ is either a double cone or a spacelike cone. Haag duality has been proven for free fields~\cite{MR0160487}, but to the knowledge of the author no result is known (in $d=2+1$) for interacting fields.

Every representation in $\bfm^{\alg{A}}$ can be described as an endomorphism of some algebra $\alg{A}^{\spc{S}_a}$ containing $\alg{A}$ as a subalgebra. The category $\bfm^{\alg{A}}$ then can be equipped with a tensor product defined by composition of such endomorphisms. As mentioned before, a particularly interesting feature is that it is in fact a braided tensor category. In three dimensions, the DHR sectors, which are localised in \emph{bounded} regions, form a degenerate tensor subcategory of $\bfm^{\alg{A}}$ with respect to the braiding: the braiding with objects from this subcategory reduces to a symmetry. By a result of Rehren, this implies that the category $\bfm^{\alg{A}}$ cannot be modular~\cite{MR1147467,MR1128146}. The basic idea now is to pass to the field net $\alg{F}$, as constructed by Doplicher and Roberts~\cite{MR1062748}. 

The \emph{field net} is a net of algebras that generate the different superselection sectors by acting on the vacuum. It is endowed with an action of a compact group $G$ of symmetries (sometimes called the gauge group). The observables are precisely those elements of the field algebra that are invariant under the action of this symmetry group. At the end of the 1980s, Doplicher and Roberts solved a long-standing problem in algebraic quantum field theory, namely how to construct the group $G$ and the corresponding field net from the observable algebra~\cite{MR1062748}. Their investigations led to a new duality theory for compact groups~\cite{MR1010160}, on which we will elaborate below. It is important to note however that these constructions only work if all sectors have \emph{permutation} statistics. In the \emph{braided} case, instead of a group one expects an object with a (quasi-)Hopf algebra-like structure, see for example~\cite{MR1155286,MR1234107}, or even a more general notion of symmetry~\cite{Kowalzig}.

In the special case where $\mc{A}$ has no fermionic DHR sectors, we can interpret $\mc{O} \mapsto \alg{F}(\mc{O})$ as a new AQFT. Conti, Doplicher and Roberts have shown that the field net does not have any non-trivial representations satisfying the DHR criterion any more~\cite{MR1828981}. The theory $\alg{F}$ is an extension of $\alg{A}$, in the sense that any stringlike localised representation of $\alg{A}$ can be extended to a representation of $\alg{F}$ with the same localisation properties. This extension factors through the categorical crossed product $\bfm^{\alg{A}} \rtimes \dhr^{\alg{A}}$ of~\cite{MR1749250}. Under certain conditions, this crossed product is in fact equivalent, in the categorical sense, to the category $\bfm^{\alg{F}}$. This makes it possible to understand the latter completely in terms of the original theory $\mc{O} \mapsto \alg{A}(\mc{O})$. To summarise, the obstruction for modularity is removed by passing from a theory $\alg{A}$ to a new theory $\alg{F}$ that extends $\alg{A}$ in a systematic way.

Although some constructions in this paper are motivated by results in $d=1+1$, there are also some notable differences with the case $d=2+1$ considered in the present work. In $d=2+1$, passing from a net $\alg{F}$ to the fixpoint theory $\alg{A} = \alg{F}^G$ with respect to the action of some group $G$ introduces DHR sectors, which are automatically degenerate in $d=2+1$. In $d=1+1$, DHR sectors also appear when passing to the fixpoint net. In this case, however, they are never degenerate, at least not if the symmetry group $G$ is finite and the theory is ``completely rational''~\cite{MR1838752}. In that situation there appear automatically ``twisted'' sectors which prevent degeneracy of the new DHR sectors in the fixpoint theory~\cite{MR2183964}.

\bigskip
The paper is organised as follows. In Section~\ref{sec:bf}, the basic structure of stringlike localised sectors in three dimensions is recalled. The next section is concerned with the construction of the field net $\alg{F}$, and it is shown that this can be interpreted as a new AQFT without DHR sectors. Section~\ref{sec:extend} then discusses how stringlike localised sectors of our original theory $\alg{A}$ can be extended to the new theory $\alg{F}$. Section~\ref{sec:restrict} deals with the reverse problem of restricting sectors that are invariant under the action of the symmetry group, using results from the theory of non-abelian cohomology. In the last part of the paper, it is investigated how these results are related to the purely mathematical theory of crossed products of braided tensor categories by symmetric subcategories. This gives a better understanding of the sectors of the new theory in terms of those of the old theory. In particular, conditions are given under which \emph{all} sectors of $\alg{F}$ are related to the sectors of $\alg{A}$. In the last section, the main results are summarised and some open problems are indicated. Some terminology regarding category theory and algebraic quantum field theory, which will be used throughout the article, is recollected in an appendix.

\section{Stringlike localised sectors}
\label{sec:bf}
In algebraic quantum field theory a superselection criterion identifies the physically relevant representations of the observable algebra. Usually one selects those representations $\pi$ that cannot be distinguished from the vacuum representation $\pi_0$ in the spacelike complement of some causally complete region. The selection criterion used by Doplicher, Haag and Roberts (DHR) requires that the relevant representations $\pi$ satisfy, for each double cone $\mathcal{O}$, 
\begin{equation}
	\label{eq:selectdhr}
  \pi \upharpoonright \alg{A}(\mc{O}') \cong \pi_0\upharpoonright\alg{A}(\mc{O}').
\end{equation}
That is, $\pi$ is unitarily equivalent to the vacuum representation when restricted to observables in the causal complement of an arbitrary double cone. The structure of the DHR superselection sectors is well understood, see e.g.~\cite{MR1405610,halvapp} for reviews. A DHR representation is of the form $\pi \cong \pi_0 \circ \rho$, where $\rho$ is an endomorphism\footnote{All (endo)morphisms and representations are assumed to be unital and to preserve the $^*$-operation, unless stated otherwise.} of $\alg{A}$ that acts trivially on $\alg{A}(\mc{O}')$ for some $\mc{O}$. Such an endomorphism is said to be \emph{localised} in $\mc{O}$. Furthermore, $\rho$ is \emph{transportable}, in the sense that for any double cone $\widehat{\mc{O}}$ there is a morphism $\widehat{\rho}$ localised in $\widehat{\mc{O}}$, unitarily equivalent to $\rho$. Localised transportable endomorphisms can be regarded as objects of a braided tensor category. 

However, the criterion~\eqref{eq:selectdhr} is too narrow for many physical applications. For example, consider the case of an electrically charged particle. Then, by Gauss' theorem, it is possible to measure the electric  flux through a surface at arbitrary large distance. This implies that the presence of an electric charge can be detected at arbitrarily large distances, i.e., there is no double cone $\mc{O}$ such that the state cannot be distinguished from the vacuum in the spacelike complement of this $\mc{O}$. See~\cite{MR667767} for a discussion of states in QED. This is why Buchholz and Fredenhagen consider a more general selection criterion~\cite{MR660538}, namely
\begin{equation}
  \label{eq:selectbf}
  \pi\upharpoonright\alg{A}(\spc{C}') \cong \pi_0\upharpoonright\alg{A}(\spc{C}'),
\end{equation}
for each \emph{spacelike cone} $\spc{C}$ in the following sense: 
\begin{definition}
	A \emph{spacelike cone} is a set $\spc{C} = x + \bigcup_{\lambda > 0} \lambda \cdot \mc{O}$, for some double cone $\mc{O}$ not containing the origin, and $x \in \mathbb{M}^d$. Moreover, we demand that $\spc{C}$ is causally complete\footnote{Buchholz and Fredenhagen do not demand that $\spc{C}$ is causally complete~\cite{MR660538}. However, in view of our definition of Haag duality, it is more natural to consider only causally complete spacelike cones. See the Appendix to~\cite{MR1062748} for an alternative, but equivalent, definition.}, i.e., $\spc{C} = \spc{C}''$.
\end{definition}
Such a spacelike cone can be visualised as a semi-infinite string that becomes thicker and thicker when moving towards spacelike infinity. Since again this criterion means that such representations cannot be distinguished from the vacuum in the spacelike complement of a spacelike cone, such representations are called \emph{localisable in cones}. We will call the equivalence class of such a representation a \emph{BF sector}, and call a representative a \emph{BF representation}.

Buchholz and Fredenhagen show that in a relativistic quantum field theory \emph{massive single-particle representations} always have such localisation properties. Roughly speaking, a \emph{massive representation} is a representation that is covariant under translation (covariance under the full Poincar\'e group is not required). Moreover, the joint spectrum of the generators of the translations is bounded away from zero and contains an isolated mass shell, separated by a gap from the rest of the spectrum.

There are several methods to study the superselection structure of charges localised in spacelike cones (also called ``topological charges''). Recall that we identified $\pi_0(\alg{A})$ with $\alg{A}$. Contrary to the case of DHR sectors, BF sectors cannot be described in terms of endomorphisms of the quasi-local algebra $\alg{A}$. Instead, the representations map cone algebras $\alg{A}(\spc{C})$ to weak closures of the algebra, that is, $\eta(\alg{A}(\spc{C})) \subset \alg{A}(\spc{C})''$ if $\eta$ is localised in a spacelike cone $\widehat{\spc{C}} \subset \spc{C}$. For double cones $\mc{O}$ there is the inclusion $\alg{A}(\mc{O})'' \subset \alg{A}$ (recall that the local algebras are assumed to be von Neumann algebras), but for spacelike cones in general the weak closure $\alg{A}(\spc{C})''$ is not contained in $\alg{A}$. This implies that BF representations do not map $\alg{A}$ into $\alg{A}$, as is the case in the DHR situation, but into some larger algebra. This situation is rather inconvenient, but fortunately this problem can be solved by introducing an auxiliary algebra~\cite{MR660538}. The BF representations can be extended to proper endomorphisms of this auxiliary algebra. At the end of this section we comment on some other approaches.

To motivate the introduction of the auxiliary algebra, consider a BF representation $\pi$ and spacelike cone $\spc{C}$. By the selection criterion~\eqref{eq:selectbf} there is a unitary $V$ such that $\pi_0(A) = V \pi(A) V^*$ for all $A \in \alg{A}(\spc{C}')$. Consider the equivalent representation
\[
\eta(A) = V \pi(A) V^{*}, \quad A \in \alg{A}.
\]
It follows that $\eta(A) = A$ for all $A \in \alg{A}(\spc{C}')$. By localisation and locality it follows that $\eta(AB) = \eta(A)B = B \eta(A)$ for all $A \in \alg{A}(\widehat{\spc{C}})$ and $B \in \alg{A}(\widehat{\spc{C}}')$ where $\widehat{\spc{C}} \supset \spc{C}$ is a spacelike cone. Therefore, invoking Haag duality~\eqref{eq:haagd} for spacelike cones we have $\eta(\alg{A}(\widehat{\spc{C}})) \subset \alg{A}(\widehat{\spc{C}})''$. 

\begin{definition}
  A representation $\eta$ of $\alg{A}$ is a \emph{BF representation localised in $\spc{C}$} if it satisfies the selection criterion~\eqref{eq:selectbf} and $\eta(A) = A$ for all $A \in \alg{A}(\spc{C}')$. This is denoted by $\eta \in \bfm^{\alg{A}}(\spc{C})$.
\end{definition}
From now on, fix a spacelike cone $\spc{C}$. We will consider the category $\bfm^{\alg{A}}(\spc{C})$ of BF representations localised in $\spc{C}$ and intertwiners\footnote{Recall that for two representations $\eta_1$ and $\eta_2$ of an algebra $\alg{A}$, an intertwiner $T$ from $\eta_1$ to $\eta_2$ is an operator such that for all $A \in \alg{A}$, $T \eta_1(A) = \eta_2(A) T$.} as morphisms. Note that the objects of the category are still transportable, i.e., if $\eta \in \bfm^{\alg{A}}(\spc{C})$ and if $\widehat{\spc{C}}$ is an arbitrary spacelike cone, there is a unitary equivalent representation (that might not be an object of $\bfm^{\alg{A}}(\spc{C})$) that is localised in $\widehat{\spc{C}}$. This restriction to a fixed spacelike cone is for technical reasons only. As will be demonstrated below, for two spacelike cones $\spc{C}_1$ and $\spc{C}_2$, the corresponding categories $\bfm^{\alg{A}}(\spc{C}_i)$ are equivalent as braided tensor categories. In the remainder of this section, the structure of this category is described. The reader unfamiliar with these constructions is advised to keep in mind the category of finite-dimensional unitary representations of a compact group, which shares many of its features with the category of BF representations. There is, however, one notable difference: the representation category of a compact group is always \emph{symmetric}, whereas the category of BF representations in $d=2+1$ is interesting precisely because it is \emph{braided}, but in general not symmetric.

We now come to the construction of the auxiliary algebra. One starts by choosing an auxiliary spacelike cone $\spc{S}_a$. This can be interpreted as a ``forbidden'' direction. From now on this auxiliary cone will be fixed. It should be noted that the results will not depend on the specific choice of $\spc{S}_a$. After fixing $\spc{S}_a$ we can consider the family of algebras $\alg{A}( (\spc{S}_a + x)')''$, for $x \in \mathbb{M}^3$. This set is partially ordered by $x \leq y \Leftrightarrow \spc{S}_a + x \supset \spc{S}_a + y$ and is directed, i.e., each pair of elements in this poset has an upper bound. Hence it is possible to consider the $C^*$-inductive limit (here the norm closure of the union of algebras)
\[
\alg{A}^{\spc{S}_a} = \overline{\bigcup_{x \in \mathbb{M}^3} \alg{A}( (\spc{S}_a + x)')''}^{\|\cdot\|} \subset \mc{B}(\mc{H}_0).
\]
Clearly for every $x \in \mathbb{M}^3$, we have $\alg{A}^{\spc{S}_a} = \alg{A}^{\spc{S}_a + x}$. The point is then that BF representations can be extended to endomorphisms of the \emph{auxiliary} algebra.

After the introduction of this auxiliary algebra, the structure of the superselection sectors can be studied with essentially the same methods as in the case of compactly localised (DHR) sectors, see e.g.~\cite{MR1405610,halvapp}. For the convenience of the reader and to establish our notation, the main features and constructions are outlined below. The results are phrased in terms of tensor $C^*$-categories. See~\cite{MR1010160,MR1444286,mmappendix} for an overview of the relevant notions. 
\begin{lemma}
	\label{lem:extend}
	Let $\eta$ be a BF representation. Then $\eta$ has a unique extension $\eta^{\spc{S}_a}$ to $\alg{A}^{\spc{S}_a}$ that agrees with $\eta$ on $\alg{A}$ and is weakly continuous on $\alg{A}( (\spc{S}_a+x)')''$ for each $x \in \mathbb{M}^3$. If $\eta$ is localised in $\spc{C} \subset (\spc{S}_a + x)'$ for some $x \in \mathbb{M}^3$, then $\eta^{\spc{S}_a}$ is an endomorphism of $\alg{A}^{\spc{S}_a}$.
	In the latter case we have $\eta_1^{\spc{S}_a} \circ \eta_2^{\spc{S}_a} = \eta_2^{\spc{S}_a} \circ \eta_1^{\spc{S}_a}$ if the localisation regions of $\eta_1$ and $\eta_2$ are spacelike separated.
\end{lemma}
\begin{proof}
  We give a sketch of the proof; for the full proof see Lemma 4.1 and Proposition 4.3 of~\cite{MR660538}. By the superselection criterion it is possible to find a unitary $V$ in $\mc{B}(\mc{H}_0)$ such that $\eta(A) = V A V^*$ for $A \in \alg{A}( (\spc{S}_a + x)')$. This representation can be extended uniquely to the weak closure $\alg{A}( (\spc{S}_a+x)')''$. Obviously, this extension is weakly continuous. This leads to an extension $\eta^{\spc{S}_a}$ of $\eta$. By Haag duality the localisation of $\eta$ implies, in particular, that the unitaries $V$ can be chosen in the auxiliary algebra, so that $\eta^{\spc{S}_a}$ is an endomorphism of this auxiliary algebra.

The final statement of the lemma can be checked for $A \in \alg{A}$. We then invoke weak continuity to arrive at the desired conclusion.
\end{proof}

With this result, the analysis of the structure of the BF representations proceeds analogously to the DHR case: one just extends the representations to $\alg{A}^{\spc{S}_a}$ as appropriate. In particular, it is possible to compose endomorphisms, which can be interpreted as composition of charges.

\begin{definition}
  \label{def:tproduct}
  Let $\eta_i \in \bfm^{\alg{A}}(\spc{C})$ ($i=1,2$), with $\spc{C}$ spacelike to $\spc{S}_a + x$ for some $x$. Define a tensor product on $\bfm^{\alg{A}}(\spc{C}))$ by 
\[
	\eta_1 \otimes \eta_2 = \eta_1^{\spc{S}_a} \circ \eta_2,
\]
and if $T_i \in \Hom_\alg{A}(\eta_i, \sigma_i)$ for $i=1,2$, by
\[
T_1 \otimes T_2 = T_1 \eta_1^{\spc{S}_a}(T_2) = \sigma_1^{\spc{S}_a}(T_2) T_1.
\]
\end{definition}
It can be shown that $\eta_1 \otimes \eta_2 \in \bfm^{\alg{A}}(\spc{C})$ and that $\eta_1 \otimes \eta_2$ is independent of the specific choice of auxiliary cone. Moreover if $\eta_i \cong \widehat{\eta}_i$, then $\eta_1 \otimes \eta_2 \cong \widehat{\eta}_1 \otimes \widehat{\eta}_2$. See Section 4 of~\cite{MR660538} for proofs. 

To proceed, an additional property is necessary, namely Borchers' Property B for spacelike cones.
\begin{property}
Let $E \in \alg{A}(\spc{C}')'$ be a non-zero projection. Then, for any spacelike cone $\widehat{\spc{C}} \supset \overline{\spc{C}}$, where the bar denotes closure in $\mathbb{M}^3$, there is an isometry $W \in \alg{A}(\widehat{\spc{C}}')'$ such that $WW^* = E$.
\end{property}
In fact, this property follows from the spectrum condition and locality~\cite{MR0211705}, or~\cite{MR1147465} for a more recent exposition. Note that the assumption of weak additivity is not necessary, since this is automatically satisfied for algebras of observables localised in spacelike cones. Moreover, if the $\alg{A}(\spc{C})''$ are Type III factors Property B is satisfied automatically and one can even choose $W \in \alg{A}(\spc{C})''$.

\begin{theorem}
	The category $\bfm^{\alg{A}}(\spc{C})$ has subobjects (notation: $\eta_1 \prec \eta_2$), direct sums $\eta_1 \oplus \eta_2$, and can be endowed with a tensor product $\eta_1 \otimes \eta_2$.
\end{theorem}
\begin{proof}
  The first two properties can be derived using Property B. First, consider $\eta \in \bfm^{\alg{A}}(\spc{C})$ and a projection $P \in \End_{\alg{A}}(\eta)$. Consider a spacelike cone $\widehat{\spc{C}} \supset \overline{\spc{C}}$. By Property B there exists an isometry $W \in \alg{A}(\widehat{\spc{C}})''$ such that $P = W W^*$. Define $\sigma(-) = W^* \eta(-) W$. Note that $W \in \Hom_{\alg{A}}(\sigma, \eta)$. By duality and the localisation of $\eta$, it follows that $\sigma$ is localised in $\widehat{\spc{C}}$. Moreover, since $\eta$ is localisable in cones it is easy to exhibit unitary charge transporters of $\sigma$, hence $\sigma \in \bfm^{\alg{A}}(\spc{C})$. By transportability it is possible to find a unitarily equivalent $\widehat{\sigma}$ localised in $\spc{C}$. It follows that $\widehat{\sigma} \prec \eta$.

For the existence of direct sums, consider $\eta_1,\eta_2 \in \bfm^{\alg{A}}$. Using again Property B it is possible to find isometries $V_1, V_2 \in \alg{A}(\widehat{\spc{C}})''$ such that $V_1 V_1^* + V_2 V_2^* = I$ (consider projections $P \neq 0, I$ and $I-P$). Define $\eta(-) = V_1 \eta_1(-) V_1^* + V_2 \eta_2(-) V_2^*$. Then $\eta$ is localised in $\widehat{\spc{C}}$ and localisable in cones. Using the same argument as above, an equivalent $\widehat{\eta}$ localised in $\spc{C}$ can be found. This is the direct sum $\eta = \eta_1 \oplus \eta_2$, unique up to isomorphism. To see this, suppose $\eta'(-) = W_1 \eta_1(-) W_1^* + W_2 \eta_2(-) W_2^*$. Then $U := V_1 W_1^* + V_2 W_2^*$ is a unitary intertwiner from $\eta$ to $\eta'$. Similarly, it is not hard to see that if $\eta \cong \eta'$, then $\eta'$ is a direct sum of $\eta_1$ and $\eta_2$ as well.

The tensor product was already defined in Definition~\ref{def:tproduct}. With these definitions it is straightforward to verify that $\otimes$ defines a bifunctor on the category, and turns $\bfm^{\alg{A}}(\spc{C})$ into a strict monoidal category, with monoidal unit $\iota$, given by the identity endomorphism of $\alg{A}$.
\end{proof}

Now that a tensor product has been defined on the category $\bfm^{\alg{A}}(\spc{C})$, the next step is to look for a \emph{braiding}. The braiding is intimately related to the statistics of a sector. It gives rise to representations of the braid group, or of the symmetric group if the braiding is symmetric, describing the interchange of identical particles. These notions were first studied in the context of algebraic quantum field theory by Doplicher, Haag and Roberts~\cite{MR0297259,MR0334742}. Braid statistics have been studied, for example, in~\cite{MR1016869}. The constructions below are essentially the same as in these original papers, and have merely been adapted to the case at hand.

A convenient technical tool when dealing with BF representations is that of an interpolating sequence of spacelike cones. This can be used, e.g., to show that a certain construction is independent of the specific choice of spacelike cones, or to choose charge transporters in the auxiliary algebra.
\begin{definition}
  Let $\spc{C}_1$ and $\spc{C}_2$ be spacelike cones in $\spc{S}_a'$. An \emph{interpolating sequence} between $\spc{C}_1$ and $\spc{C}_2$, is a set of spacelike cones $\widehat{\spc{C}}_1, \dots \widehat{\spc{C}}_n$, each contained in $(\spc{S}_a+x_i)'$ for some $x_i \in \mathbb{M}^3$, such that $\widehat{\spc{C}}_1 = \spc{C}_1$, $\widehat{\spc{C}}_n = \spc{C}_2$, and for each $i$ we have either $\spc{C}_i \subset \spc{C}_{i+1}$ or $\spc{C}_{i+1} \subset \spc{C}_i$.
\end{definition}

With this definition it is possible to prove the following result:
\begin{lemma}
  \label{lem:chargetr}
  Let $\eta \in \bfm^{\alg{A}}(\spc{C}_1)$. For any spacelike cone $\spc{C}_2 \subset \spc{S}_a'$ there is an equivalent representation $\widehat{\eta} \cong \eta$ localised in $\spc{C}_2$, such that a unitary intertwiner $V$ in $\alg{A}^{\spc{S}_a}$ can be found.
\end{lemma}
\begin{proof}
Choose an interpolating sequence $\widehat{\spc{C}}_i$ between $\spc{C}_1$ and $\spc{C}_2$. Set $\widehat{\eta}_1 = \eta$. We then define a sequence of unitarily equivalent representations $\widehat{\eta}_{i+1} \cong \widehat{\eta}_{i}$, such that $V_i \widehat{\eta}_{i+1} = \widehat{\eta}_i V_i$. Since either $\spc{C}_{i+1} \subset \spc{C}_i$ or $\spc{C}_i \subset \spc{C}_{i+1}$, it follows by Haag duality that either $V_i \in \alg{A}(\spc{C}_i)''$ or $V_i \in \alg{A}(\spc{C}_{i+1})''$, hence $V_i \in \alg{A}^{\spc{S}_a}$. But then $V_{n-1} \cdots V_1$ is a unitary intertwiner between $\widehat{\eta} \equiv \widehat{\eta}_n$, and because $\alg{A}^{\spc{S}_a}$ is an algebra, it follows that $V \equiv V_{n-1} \cdots V_1 \in \alg{A}^{\spc{S}_a}$.
\end{proof}

A \emph{braiding} on the category relates the objects $\eta_1 \otimes \eta_2$ and $\eta_2 \otimes \eta_1$. In this case it is a unitary operator $\varepsilon_{\eta_1,\eta_2}$ that intertwines the representations $\eta_1 \otimes \eta_2$ and $\eta_2 \otimes \eta_1$. A particular example is the \emph{statistics operator} $\varepsilon_{\eta,\eta}$ that describes the statistics of a sector. To define the braiding $\varepsilon_{\eta_1, \eta_2}$ between $\eta_1 \otimes \eta_2$ and $\eta_2 \otimes \eta_1$, with $\eta_i \in \bfm^{\alg{A}}(\spc{C})$, first choose two spacelike cones $\widehat{\spc{C}}_1$ and $\widehat{\spc{C}}_2$. Both spacelike cones should lie in the causal complement of $\spc{S}_a+x$ for some $x$ and should lie spacelike with respect to each other, i.e. $\widehat{\spc{C}}_1 \subset \widehat{\spc{C}}_2'$. By transportability there are BF-representations $\widehat{\eta}_i \cong \eta_i$ localised in $\widehat{\spc{C}}_i$. These morphisms are called \emph{spectator morphisms}. Moreover, by Lemma~\ref{lem:chargetr} the corresponding unitary intertwiners $V_1 \in \Hom_{\alg{A}}(\eta_1, \widehat{\eta}_1)$ and $V_2$ can be chosen to be in $\alg{A}^{\spc{S}_a}$. After these choices have been made, one can define the braiding by 
\[
\varepsilon_{\eta_1, \eta_2} = (V_2 \otimes V_1)^* \circ (V_1 \otimes V_2).
\]
It follows that $\varepsilon_{\eta_1, \eta_2}$ is a unitary in $\Hom_{\alg{A}}(\eta_1 \otimes \eta_2, \eta_2 \otimes \eta_1)$.

A standard argument using interpolating sequences of spacelike cones shows that the definition of $\varepsilon_{\eta_1, \eta_2}$ is independent of the specific choice of intertwiners and localisation regions, up to the relative position of $\spc{C}_1$ and $\spc{C}_2$, in the following sense.
\begin{definition}
  Suppose we have a spacelike cone $\spc{C}$ in the causal complement of $\spc{S}_a$. If we rotate the spatial coordinates counter-clockwise, at some point it will fail to be spacelike to $\spc{S}_a$. Now suppose we have two spacelike separated cones $\spc{C}_1$ and $\spc{C}_2$. We define an \emph{orientation} $\spc{C}_1 < \spc{C}_2$ if and only if we can move $\spc{C}_1$ by translation and rotating counter-clockwise to $\spc{S}_a$ while remaining in the spacelike complement of $\spc{C}_2$. Note that for any two spacelike separated cones, there is always precisely one cone for which this is possible.
\end{definition}
We will always choose $\widehat{\spc{C}}_2 < \widehat{\spc{C}}_1$ to define the braiding $\varepsilon_{\eta_1, \eta_2}$. One can then show that $\varepsilon_{\eta_1, \eta_2}$ is \emph{natural}, in the categorical sense, in both the first and second variable. Moreover, $\varepsilon_{\eta_1,\eta_2}$ satisfies the braid relations. The verification becomes straightforward if one chooses the spacelike cones $\widehat{\spc{C}}_i$ in the definition in a convenient way, so as to be able to make use of the localisation properties of the endomorphisms. See~\cite{halvapp} for the way this works in the DHR case.

\begin{theorem}
 The category $\bfm^{\alg{A}}(\spc{C})$ is a strict \emph{braided tensor category}, where the braiding is given by $\varepsilon_{\eta_1, \eta_2}$.
\end{theorem}
The appearance of braid (but not symmetric) statistics is due to the fact that in 2+1 dimensions the manifold of spacelike directions is not simply connected, unlike the situation in higher dimensions. See Section 2 of~\cite{MR2472029} for a clarification of this point.

Finally, there is the categorical notion of a \emph{conjugate object}. In this setting, conjugates can be interpreted as ``anti-particles'', and are closely related to the statistics of a sector. To each BF representation $\eta$ a dimension $d(\eta)$ and phase $\omega_\eta$ are associated. For bosons (resp. fermions) the phase is $+1$ (resp. $-1$), but in $\bfm^{\alg{A}}(\spc{C})$ these are not the only possibilities (for $d=2+1$). There are several ways to introduce these parameters. The traditional way is to introduce a left inverse~\cite{MR0297259,MR0334742}. Longo discovered a connection between the Jones index of an inclusion of factors and the dimension~\cite{MR1027496,MR1059320}. Finally, one can define the dimension, and twist (or phase), in a general categorical setting~\cite{MR1444286}, see also~\cite{mmappendix}.

The dimension $d(\eta)$ takes values in $[1,\infty]$. If $d(\eta) < \infty$, one says that $\eta$ has \emph{finite statistics}. Restricted to objects of finite dimension, the dimension function satisfies the following identies:
\[
d(\overline{\eta}) = d(\eta), \quad d(\eta_1 \otimes \eta_2) = d(\eta_1) d(\eta_2), \quad d(\eta_1 \oplus \eta_2) = d(\eta_1) + d(\eta_2).
\]
Here $\overline{\eta}$ is a conjugate representation of $\eta$ (see the Appendix). From now on, we will consider only categories where all objects have finite dimension, i.e., we leave out any sectors with infinite statistics the observable net may admit. Objects with finite dimension are precisely those for which there is a conjugate (or ``anti-particle''). To avoid cumbersome notation, the category of all BF representations with \emph{finite} statistical dimension will also be denoted by $\bfm^{\alg{A}}(\spc{C})$.

Under weak additional assumptions, Guido and Longo showed that the DHR sectors with finite statistics are Poincar\'e covariant with positive energy~\cite{MR1181069}, in particular they are covariant under translations as well. Hence under their assumptions, the set of finite DHR sectors coincides with the set of Poincar\'e covariant finite sectors with positive energy. Moreover, Buchholz and Fredenhagen show that massive irreducible single particle representations automatically have finite statistics~\cite{MR660538}. They also show that all representations of interest for particle physics are indeed described by (direct sums of) representations with finite statistics. One may therefore argue that restricting to sectors of finite dimension is not too restrictive from the point of view of physics. Finally, we would like to mention that Mund recently proved a version of the spin-statistics theorem for massive particles obeying braid group statistics~\cite{MR2472029}.

The restriction to sectors with finite statistics implies that the category $\bfm^{\alg{A}}(\spc{C})$ is semi-simple, i.e. that every representation can be decomposed into a direct sum of irreducibles. Indeed, let $\eta \in \bfm^{\alg{A}}(\spc{C})$. If $\eta$ is not irreducible there is a non-trivial projection $E \in \End_{\alg{A}}(\eta)$. By the existence of subobjects, one has $\eta = \eta_1 \oplus \eta_2$ for some $\eta_1, \eta_2 \in \bfm^{\alg{A}}(\spc{C})$. Semi-simplicity now follows, since $d(\eta) = d(\eta_1) + d(\eta_2)$ and the dimension function $d$ takes values in $[1,\infty)$, since we restricted to objects of finite dimension.

The results so far can be summarised by the following theorem.
\begin{theorem}
  The category $\bfm^{\alg{A}}(\spc{C})$ is a braided tensor $C^*$-category. That is it has duals (or conjugates), direct sums, subobjects, a braiding and a positive $*$-operation. The Hom-sets are Banach spaces, such that $\| T\circ S\| \leq \| S \| \|T\|$ and $\| S^* \circ S \| = \|S\|^2$ for all morphisms $S,T$ (whenever the composition is defined). Moreover, the tensor unit $\iota$ is irreducible.
\end{theorem}
It then follows automatically that the Hom-sets are finite-dimensional vector spaces~\cite{MR1444286}. In the case of interest here, the $*$-operation and norm are inherited from the observable algebra.

One question that remains to be answered is to which extent the category $\bfm^{\alg{A}}(\spc{C})$ depends on the choice of $\spc{C}$. It turns out that in fact for any two choices $\spc{C}_1, \spc{C}_2$ the resulting categories are equivalent as tensor categories, c.f.~\cite[Theorem 4.11]{MR1062748}.
\begin{prop}
  \label{prop:equiv}
  Let $\spc{C}_1$ and $\spc{C}_2$ be two spacelike cones. Then the categories $\bfm^{\alg{A}}(\spc{C}_1)$ and $\bfm^{\alg{A}}(\spc{C}_2)$ are equivalent as braided tensor categories.
\end{prop}
\begin{proof}
  We give a sketch of the proof; the details are left to the reader. One first proves the result in the case $\spc{C}_1 \subset \spc{C}_2$. This gives rise to a full and faithful inclusion of categories $\bfm^{\alg{A}}(\spc{C}_1) \subset \bfm^{\alg{A}}(\spc{C}_2)$. Clearly this inclusion is braided. In addition, the inclusion is essentially surjective, since for each representation localised in $\spc{C}_2$ one can find a unitary equivalent representation localised in $\spc{C}_1$. Hence, the inclusion is in fact an equivalence of categories, hence an equivalence of braided tensor categories~\cite{MR0338002}.
 
  To prove the full result, one uses an argument with interpolating sequences of spacelike cones.
\end{proof}

Thus the BF representations form a braided tensor category. However, if there are DHR localised sectors, the braiding has a ``trivial'' part. Indeed, the DHR sectors form a symmetric subcategory of $\bfm^{\alg{A}}(\spc{C})$. But more importantly, the DHR sectors are \emph{degenerate} objects with respect to the braiding. That is, they have trivial braiding with \emph{any} object of $\bfm^{\alg{A}}(\spc{C})$, in a sense made precise below. In such a situation, one says that the category has a non-trivial centre~\cite{MR1749250}.
\begin{definition}
 The \emph{centre} of a braided category is the full subcategory of degenerate objects. That is, it consists of all objects $\rho$ such that $\varepsilon_{\rho,\eta} \circ \varepsilon_{\eta, \rho} = I_{\eta \otimes \rho}$ for all objects $\eta$.
\end{definition}
A non-trivial centre is an obstruction to modularity, since by a result of Rehren the existence of (non-trivial) degenerate sectors implies that the so-called $S$-matrix (in the sense of Verlinde~\cite{MR954762}) is not invertible~\cite{MR1147467}. To make this situation more precise, we study the properties of the DHR sectors within $\bfm^{\alg{A}}(\spc{C})$.

\begin{definition}
  Let $\spc{S}$ be either a double cone or a spacelike cone. We write $\dhr^{\alg{A}}(\spc{S})$ for the category of DHR localised sectors whose localisation region lies in $\spc{S}$.
\end{definition}
Note that $\rho \in \dhr(\spc{C})$ in particular is also an element of $\bfm(\spc{C})$, so the constructions in the first part of this section go through without change. For example, the tensor product of $\rho_1$ and $\rho_2$ in $\dhr(\spc{C})$ is again in $\dhr(\spc{C})$. Since objects from $\dhr(\spc{C})$ can be localised in \emph{bounded} regions of spacetime, one can say even more about them:
\begin{figure}
	\includegraphics[width=8cm]{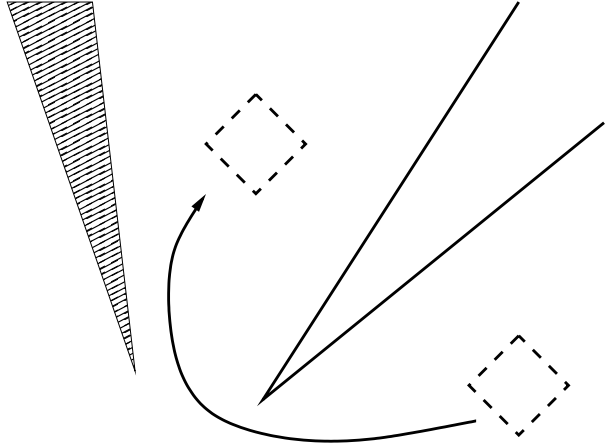}
	\caption{This figure shows why the braiding is degenerate for compactly localised endomorphisms. The compactly localised (dashed lines) endomorphism can move from one side of the spacelike cone to the other, keeping it in the causal complement of the auxiliary cone (shaded region) and spacelike cone $\spc{C}$ (solid lines) at all times.}
	\label{fig:braiddeg}
\end{figure}
\begin{lemma}
  Let $\eta \in \bfm^{\alg{A}}(\spc{C})$ and $\rho \in \dhr^{\alg{A}}(\mathcal{O})$ for some double cone $\mc{O} \subset \spc{S}_a'$. Then the DHR sectors are \emph{degenerate} with respect to the braiding, i.e.,
 \[
	\varepsilon_{\rho,\eta} \circ \varepsilon_{\eta, \rho} = I_{\eta \otimes \rho}.
 \]
 \label{lem:trivmon}
\end{lemma}
\begin{proof}
The basic idea is depicted in Figure~\ref{fig:braiddeg}. Because $\rho$ is localised in a bounded region, there is more freedom in the choice of localisation cones of the spectator morphisms. In particular, it is possible to ``flip'' the cones, that is, if $\hat{\rho}$ is localised in some spacelike cone $\widehat{\spc{C}}$, it is possible to find a spacelike cone $\widetilde{\spc{C}}$ pointing in the opposite direction, such that $\hat{\rho}$ is localised in $\widetilde{\spc{C}}$. Using this, it is not difficult to see that the braiding $\varepsilon_{\rho,\eta}$ does not depend on the orientation of the spacelike cones of the spectator morphisms. It follows that $\varepsilon_{\rho,\eta} = \varepsilon_{\eta,\rho}^{-1}$, which proves the result.
\end{proof}

To conclude this section we briefly comment on other methods to describe the superselection structure of charges localised in spacelike cones. Doplicher and Roberts take a different approach in~\cite{MR1062748}, which does not need the auxiliary algebra. This method, however, works only in spacetimes of dimension at least 4 and would need adaptation to the $d=2+1$ case we are interested in. 

In the approach of both Buchholz \& Fredenhagen and of Doplicher \& Roberts, only representations localised in a \emph{fixed} spacelike cone $\spc{C}$ can be considered. A related approach by Fr\"ohlich and Gabbiani~\cite{MR1104414}, which also uses the auxiliary algebra, does not require one to fix a spacelike cone. Instead, they consider two coordinate patches, and show that it is possible to pass from one to the other in a ``smooth'' way.

Finally, it is possible to use the so-called \emph{universal algebra}, introduced by Fredenhagen~\cite{MR1147469}, see also~\cite{Mund2009}. This has the advantage that we do not have to choose an auxiliary cone. On the other hand, there are drawbacks, for example the universal algebra is not simple and the vacuum representation is not faithful~\cite{MR1199171}. In the end, each method gives the same result, so the choice of method only matters for the technical details.

\section{The field net}
\label{sec:field}
In this section we consider the field net of the observable algebras \emph{with respect to the DHR sectors}. In other words, the field operators by construction only generate the DHR sectors. This is possible since the DHR sectors have permutation statistics in 2+1 dimensions. At the end of this section we discuss an alternative, more abstract construction of the field net, that turns out to be helpful in the applications we have in mind.

For the convenience of the reader we first recall the definition of a \emph{field net}~\cite{MR1062748}. We specialise to the case of interest here: that of a complete, normal field net without fermionic sectors.
\begin{definition}
	\label{def:fieldnet}
  Let $(\pi_0, \mc{H}_0)$ be a vacuum representation of the net $\mc{O} \mapsto \alg{A}(\mc{O})$. A \emph{complete normal field net} $(\pi, G, \alg{F})$ is a representation $(\pi, \mc{H})$ of $\alg{A}$ and a net $\mc{O} \mapsto \alg{F}(\mc{O})$ of von Neumann algebras acting on $\mc{H}$, such that
\begin{enumerate}
	\item $\mc{H}_0 \subset \mc{H}$;
	\item $\pi_0$ is a subrepresentation of $\pi$;
	\item there is a compact group $G$ of unitaries on $\mc{H}$ leaving $\mc{H}_0$ pointwise fixed, inducing an action $\alpha_g = \Ad g$;
	\item for each $g \in G$, $\alpha_g$ is an automorphism of $\alg{F}(\mc{O})$ with fixed-point algebra $\pi(\alg{A}(\mc{O}))$;
	\item the inductive limit $\alg{F}$ of the local algebras $\alg{F}(\mc{O})$ is irreducible;
	\item the Hilbert space $\mc{H}_0$ is cyclic for $\alg{F}(\mc{O})$;
	\item \label{it:commut}if $\mc{O}_1$ and $\mc{O}_2$ are spacelike separated double cones, $\alg{F}(\mc{O}_1)$ and $\alg{F}(\mc{O}_2)$ commute; 
	\item every irreducible DHR representation with finite statistics is included as a subrepresentation of $\pi$.
  \end{enumerate}
\end{definition}
In the presence of fermionic sectors, item~\eqref{it:commut} has to be modified to graded commutativity. Doplicher and Roberts show that such a field net exists and is unique up to a suitable notion of equivalence. The main point for us is that (at least in the purely bosonic case) this field net can be interpreted as an algebraic quantum field theory in its own right. The proof of this fact will be given below, after some preparatory results on harmonic analysis on the field net.
\begin{definition}
	Let $\xi$ be a finite-dimensional continuous unitary representation of a group $G$ as in Definition~\ref{def:fieldnet}. A set of operators $X_1, \dots X_d$, where $d = \dim \xi$, is said to be a \emph{multiplet transforming according to $\xi$} if
\[
	\alpha_g(X_i) = \sum_{j=1}^d u_{ji}^{\xi}(g) X_j,
\]
where $u_{ji}^{\xi}(g)$ are the matrix coefficients of $\xi$. An operator $X$ is said to \emph{transform irreducibly according to $\xi$}, or to be an \emph{irreducible tensor}, if it is part of a multiplet transforming according to an irreducible representation $\xi$.
\end{definition}
Irreducible tensors can be obtained by averaging over the symmetry group $G$, and their span is weakly dense in the field algebra, see e.g.~\cite[Section 2]{MR0325053}.

Recall that for each irreducible DHR endomorphism $\rho$ there is a Hilbert space $H_\rho$ in the field net transforming according to some irrep $\xi$ of $G$. That is, $H_\rho$ is a closed linear subspace of $\alg{F}$ such that $\psi_1^* \psi_2 \in \mathbb{C} I$ for all $\psi_1, \psi_2 \in H_\rho$. The space $H_\rho$ is precisely the set of operators $\psi$ in $\alg{F}$ such that $\psi A = \rho(A) \psi$ for all $A \in \alg{A}$, and $\alpha\upharpoonright_{H_\rho} = \xi$. Moreover, there is a basis of $H_\rho$ that is a multiplet transforming according to $\xi$. Irreducible tensors may then be decomposed into a $G$-invariant part and an operator in $H_\rho$, in the following sense:
\begin{lemma}
	\label{lem:irrdecompose}
	Let $\alg{B} \subset \mc{B}(\mc{H})$ be a $^*$-algebra, such that $\alg{F}(\mc{O}) \subset \alg{B}$ for some double cone $\mc{O}$. Suppose that $X$ transforms irreducibly under the action of $G$, that is, is contained in a finite dimensional Hilbert space transforming according to an irrep of $G$. Then there is a $B \in \alg{B} \cap G'$ and a $\psi \in H_\rho \subset \alg{F}(\mc{O})$ such that
\[
	X = B \psi,
\]
where $\psi$ transforms according to the same irreducible representation as $X$.
\end{lemma}
This decomposition is not unique, but depends on the specific choice of $H_\rho$.
\begin{proof}
	Complete $X$ to a multiplet $X_1, \dots X_d$. Without loss of generality, assume $X = X_1$. Let $\xi$ denote the representation according to which $X$ transforms. Since the field net has full spectrum, there is a Hilbert space $H_\rho$ in $\alg{F}(\mc{O})$, such that $H_\rho$ transforms according to $\xi$. Note that the equivalence class of $\rho$ corresponds to the class of the representation $\xi$. If $u_{ji}^\xi$ are the matrix coefficients describing the transformation of the multiplet, it is possible to choose an orthonormal basis $\psi_i$ of $H_\rho$ such that $\alpha_g(\psi_i) = \sum_{j=1}^d u_{ji}^{\xi}(g) \psi_j$. Now define
\[
	B = \sum_{i=1}^d X_i \psi_i^*.
\]
Since $\xi$ is a unitary representation, it follows that $\alpha_g(B) = B$, i.e. $B \in \alg{B} \cap G'$. Moreover, taking $\psi = \psi_1$, it follows that $B \psi = X_1 = X$.
\end{proof}

Now that we have the field net $\alg{F}$ at hand, it is possible to construct an auxiliary algebra with respect to $\alg{F}$, analogous to the one defined in terms of the algebra of observables $\alg{A}$. Hence we define
\[
\alg{F}^{\spc{S}_a} = \overline{\bigcup_{x \in \mathbb{M}^3} (\alg{F}( (\spc{S}_a + x)'))''}^{\|\cdot\|},
\]
where the closure in norm is taken. 

Since the observable net embeds into the field net, one expects the auxiliary algebra of the observable net to embed into the auxiliary algebra of the field net. The next lemma demonstrates that this is indeed the case.
\begin{lemma} Let $(\pi, G, \alg{F})$ be a complete normal field net for $(\alg{A}, \omega_0)$. Then the representation $(\pi, \mc{H})$ of $\alg{A}$ can be uniquely extended to a faithful representation $\pi^{\spc{S}_a}: \alg{A}^{\spc{S}_a} \to \mc{B}(\mc{H})$ that is weakly continuous on $\alg{A}( (\spc{S}_a+x)')''$.
\end{lemma}
\begin{proof}
  Write $\widehat{G}$ for the set of equivalence classes of irreducible representations of the group $G$. The representation $(\pi, \mc{H})$, viewed as a representation of $\alg{A}$, is a direct sum $\oplus_{\xi \in \hat{G}} d_\xi \pi_\xi$, where each $\pi_\xi$ is a DHR representation~\cite{MR1062748}. We will extend each $\pi_\xi$ to a representation $\pi^{\spc{S}_a}_\xi$ of $\alg{A}^{\spc{S}_a}$, and set $\pi^{\spc{S}_a} = \oplus_{\xi \in \hat{G}} d_\xi \pi_\xi^{\spc{S}_a}$. So consider such a representation $\pi_\xi$. By Lemma~\ref{lem:extend}, $\pi_\xi$ has a unique weakly continuous extension. In fact, since $\pi_\xi$ is localised in a bounded region, it follows in particular that $\pi^{\spc{S}_a}_\xi$ is an endomorphism of $\alg{A}^{\spc{S}_a}$, viewed as a subalgebra of $\alg{B}(\mc{H})$.
	
	To see that $\pi^{\spc{S}_a}$ is faithful, construct a left inverse $\varphi$ of $\pi^{\spc{S}_a}$, as in~\cite{MR660538}.
\end{proof}
This result makes it possible to identify $\alg{A}^{\spc{S}_a}$ with the subalgebra $\pi^{\spc{S}_a}(\alg{A}^{\spc{S}_a})$ of $\mc{B}(\mc{H})$. When there is no risk of confusion, we will sometimes identify $A \in \alg{A}^{\spc{S}_a}$ with its image $\pi^{\spc{S}_a}(A)$.

It is fruitful to investigate the relationship between the auxiliary algebra and the action of the symmetry group. Just as the observable net consists of precisely those operators that are fixed by the $G$-action on the field net, the same is true for the auxiliary algebras.
\begin{lemma}
  Let $(\pi, \mc{H}, \alg{F}, G)$ be a normal field net. Then:
  \begin{enumerate}
	\item For each spacelike cone, $\alg{F}(\spc{C})' \cap G' = \pi(\alg{A}(\spc{C}'))''$.
	\item The fixpoint algebra is given by $\left(\alg{F}^{\spc{S}_a}\right)^G = \pi^{\spc{S}_a}(\alg{A}^{\spc{S}_a})$.
  \end{enumerate}
  \label{lem:fieldfix}
\end{lemma}
\begin{proof}
(i)  First of all, since $\pi(\alg{A})'' = G'$ and $\alg{A}(\spc{C}')$ is a subalgebra of $\alg{A}$, it is obvious that $\pi(\alg{A}(\spc{C}'))'' \subseteq G'$. From relative locality, $\pi(\alg{A}(\spc{C}')) \subseteq \alg{F}(\spc{C})'$. By taking double commutants, $\pi(\alg{A}(\spc{C}'))'' \subseteq \alg{F}(\spc{C})'$.

Note that for each double cone $\mc{O}$, $\mc{H}_0$ is cyclic for $\alg{F}(\mc{O})$, hence also for $\alg{F}(\spc{C})$. This implies that an element $T \in \alg{F}(\spc{C})' \cap G'$ is uniquely determined by its restriction to $\mc{H}_0$. Furthermore, $\mc{H}_0$ is an invariant subspace for $T$, since $T \in G'$. We have $\alg{F}(\spc{C})' \cap G' \subseteq \pi(\alg{A}(\spc{C}))'$, so if $E_0$ denotes the projection onto $\mc{H}_0 \subset \mc{H}$, it follows that
\[
	T |_{\mc{H}_0} \in \pi(\alg{A}(\spc{C}))' E_0 = \pi_0(\alg{A}(\spc{C}))' = \pi_0(\alg{A}(\spc{C}'))''.
\]
The last step follows by Haag duality for spacelike cones in the vacuum representation.

(ii) Note that $\alpha_g$ extends to $\mc{B}(\mc{H})$, where $\mc{H}$ is the Hilbert space on which $\alg{F}$ acts irreducibly. Using the Haar measure of $G$, one can define a conditional expectation $\mathcal{E}: \alg{F} \to \alg{A}$ by
\[
\mathcal{E}(A) = \int_G \alpha_g(A) dg.
\]
It then follows that
\[
\mc{E}\left(\alg{F}^{\spc{S}_a}\right) = 
\overline{\mc{E}\left(\bigcup_{x \in \mathbb{M}^3} \alg{F}(\spc{S}_a+x)'\right)}^{\norm{\cdot}} = 
  \overline{\bigcup_{x \in \mathbb{M}^3} \mathcal{E}(\alg{F}(\spc{S}_a+x)')}^{\norm{\cdot}}.  
\]
where we used that $\mathcal{E}$ is weak- and norm-continuous~\cite{MR0258394}. Now by part (i) it follows that $\mc{E}(\alg{F}(\spc{S}_a+x)') = \pi^{\spc{S}_a}( \alg{A}(\spc{S}_a+x)')''$, see also~\cite[Lemma 3.2]{MR0258394}. Therefore,
\[
\mc{E}(\alg{F}^{\spc{S}_a})= \overline{\bigcup_x \pi^{\spc{S}_a}( \alg{A} (\spc{S}_a+x)' )''}^{\norm{\cdot}} = 
  \pi^{\spc{S}_a}(\alg{A}^{\spc{S}_a}),
  \]
which proves the claim.
\end{proof}

With the aid of these lemmas it is possible to prove the main result of this section: without fermionic sectors, the field net can be interpreted as an AQFT in its own right, but one without non-trivial DHR sectors. 
\begin{theorem}
	\label{thm:trivdhr}
	Assume that  $\mc{O} \mapsto \alg{A}(\mc{O})$ satisfies the following conditions:
	\begin{enumerate}
		\item there are at most countably many DHR sectors;
		\item there are no fermionic DHR sectors;
		\item each DHR sector with finite statistics is covariant under translations satisfying the spectrum condition.
	\end{enumerate}
	Then the field net $\mc{O} \mapsto \alg{F}(\mc{O})$ satisfies the axioms of an algebraic QFT, i.e. it is a local, translation covariant net satisfying Haag duality and the spectrum condition, hence it also has Property B for spacelike cones. The complete normal field net admits only the trivial DHR representation. 
\end{theorem}
\begin{proof}
	Isotony follows, since the field net is, in particular, a net. Since we assumed the absence of fermionic sectors, twisted duality for the field net reduces to Haag duality for double cones. Thus only the questions of translation covariance and duality for spacelike cones remain. The covariance properties follow from the results in Section 6 of~\cite{MR1062748}, and the assumption that we only have translation covariant sectors. In fact, one can show in this case that the representation $\pi$ of $\alg{F}$ is translation covariant. The generators of translations again satisfy the spectrum condition and the vacuum vector $\Omega$ is invariant under the action of the translation group~\cite[Section 6]{MR1062748}. By the same reasoning as before, Property B follows.
  
  To prove duality for spacelike cones, consider such a cone $\spc{C}$. First, note that by locality $\alg{F}(\spc{C}')'' \subset \alg{F}(\spc{C})'$. Let $F \in \alg{F}(\spc{C})'$ transform irreducibly under the action of $G$. But then by Lemma~\ref{lem:irrdecompose}, $F = B \psi$, where $B \in \alg{F}(\spc{C})' \cap G'$ and $\psi \in H_\rho$. Applying Lemma~\ref{lem:fieldfix} gives $B \in \pi(\alg{A}(\spc{C}'))''$ and, since $H_\rho \subset \alg{F}(\spc{C}')$, one obtains $F \in \alg{F}(\spc{C}')''$. The irreducible tensors form a dense subset, which allows us to conclude $\alg{F}(\spc{C}')'' = \alg{F}(\spc{C})'$. Taking commutants then proves Haag duality.

For the last assertion, note that the observable net is embedded in the field net. More precisely, we have an inclusion of subsystems $\alg{A} \subset \alg{F}$. By~\cite[Theorem 4.7]{MR1828981}, every DHR representation of the field net $\alg{F}$ with finite statistics is a direct sum of representations with finite statistics. Moreover, these sectors are labelled by the equivalence classes of irreducible representations of a compact group $L$, such that $\alg{F}(\alg{A})^L = \alg{B}$ (see also~\cite[Theorem 4.1]{MR1828981}). But in this case, $\alg{B} = \alg{F}(\alg{A}) = \alg{F}$, hence $L$ is the trivial group and the only irreducible DHR sector is the vacuum sector.
\end{proof}
Let us briefly comment on the assumptions of Theorem~\ref{thm:trivdhr}. The first condition is a technical one, needed for the results in~\cite{MR1828981} and Corollary~\ref{cor:restrict} below. By construction of the field net, DHR sectors are in 1-1 correspondence with irreps of $G$, hence $\widehat{G}$, the set of irreps of $G$, is also countable. The second condition implies that the field net satisfies ordinary locality, as opposed to twisted locality. The final condition is needed to lift the translation covariance of $\alg{A}$ to the field net. As mentioned before, by weak additional assumptions on $\alg{A}$, it follows automatically that every DHR sector with finite statistics is translation covariant. Therefore, the conditions appear not to be unreasonably restrictive. From now on, we will assume that $\alg{A}$ satisfies all assumptions in the theorem.

Roughly speaking, Doplicher and Roberts construct the field net as a crossed product of the observable algebras by a semigroup of endomorphisms. As mentioned before, this construction is intimately related to the theory of representations of compact groups. It is therefore not surprising that an alternative construction, based on results on the category of representations of compact groups, exists. Indeed, based on an unpublished manuscript of Roberts and on Deligne's embedding theorem~\cite{MR1106898}, Halvorson and M\"uger describe such a construction~\cite{halvapp,mmappendix}, which is of a more algebraic nature compared to the original analytic approach. Since the algebraic formulation is easier to work with in the present case, the rest of this section will be used to outline the main features of this approach and to fix the notation.

The results in Section~\ref{sec:bf} state that the DHR representations form a symmetric tensor ($C^*$)-category. By Deligne's embedding theorem, this gives rise to a faithful symmetric tensor $^*$-functor $E: \dhr^{\alg{A}} \to \mc{SH}_f$, the category of finite-dimensional (super) Hilbert spaces. The embedding theorem also gives a compact supergroup $(G,k)$ of natural monoidal transformations of $E$, and an equivalence of categories such that $\dhr^{\alg{A}}$ is equivalent to $\operatorname{Rep}_f(G,k)$. All monoidal categories and functors are assumed to be strict, unless noted otherwise.  The ``super'' structure gives a $\mathbb{Z}_2$-grading on the Hilbert spaces, corresponding to the action of a central element $k \in G$ such that $k^2 = e$. Since we assumed that all DHR sectors are bosonic, we can forget about the super structure. The group $G$ from the embedding theorem will be the symmetry group.

The embedding functor $E$ associates to each DHR endomorphism $\rho$ a Hilbert space $E(\rho)$. Using this embedding functor $E$, we first construct a field algebra $\alg{F}_0$. We cite the definition:
\begin{definition}
  The field algebra $\alg{F}_0$ consists of triples $(A, \rho, \psi)$, where $A \in \alg{A}$, $\rho \in \dhr$, and $\psi \in E(\rho)$, modulo the equivalence relation
\[
(AT, \rho, \psi) \equiv (A, \rho', E(T)\psi),
\]
for $T$ an intertwiner from $\rho$ to $\rho'$. For $\lambda \in \mathbb{C}$, we have $E(\lambda \id_\rho) = \lambda \id_{E(\rho)}$, hence $(\lambda A, \rho, \psi) = (A, \rho, \lambda\psi)$.
\end{definition}
In particular, it follows that any element with $\psi = 0$, is the zero element of the algebra. One then proceeds by defining a complex-linear structure on this algebra, a multiplication, as well as an involutive $^*$-operation. The multiplication is defined by $(A_1, \rho_1, \psi_1)(A_2, \rho_2, \psi_2) = (A_1 \rho_1(A_2), \rho_1 \otimes \rho_2, \psi_1 \otimes \psi_2)$.

The definition of the $*$-operation is a bit more involved. First, if $H$ and $H'$ are two Hilbert spaces and $S: H \otimes H' \to \mathbb{C}$ is a bounded linear map, one can define an anti-linear map $\mc{J}S: H \to H'$. This map is defined by setting $\langle (\mc{J}S)\psi, \psi'\rangle = S(\psi \otimes \psi')$ for all $\psi \in H, \psi' \in H'$. The brackets denote the inner product on $H'$. If $\rho$ is a DHR endomorphism, choose a conjugate (see the Appendix) $(\overline{\rho}, R, \overline{R})$. The $*$-operation is then defined by $(A, \rho, \psi)^* = (R^* \overline{\rho}(A)^*, \overline{\rho}, \mc{J}E(\overline{R}^*)\psi)$. For a verification that this is well-defined and indeed defines a $*$-algebra, see~\cite{halvapp}. 

Note that this construction is purely algebraic, for instance, there is no norm defined on $\alg{F}_0$. The algebra $\alg{A}$ can be identified with the subalgebra $\{(A, \iota, 1): A \in \alg{A}\}$ of $\alg{F}_0$, and $E(\rho)$ can be identified with the subspace $\{(I, \rho, \psi) : \psi \in E(\rho)\}$.\footnote{These Hilbert spaces $E(\rho)$ play the same role as the Hilbert spaces $H_\rho$ in~\cite{MR1062748}.}

The compact group $G$ associated with the embedding functor $E$ gives rise to an action on $\alg{F}_0$. Recall that the elements of $G$ are monoidal natural transformations of the functor $E$. If $g \in G$, write $g_\rho$ for the component at $\rho$. The action of $G$ on $\alg{F}_0$ is then defined by
\[
	\alpha_g(A, \rho, \psi) = (A, \rho, g_\rho \psi), \quad\quad A \in \alg{A},\quad \psi \in E(\rho).
\]
This is in fact a group isomorphism $g \mapsto \alpha_g$ into $\operatorname{Aut}_\alg{A}(\alg{F}_0)$, the group of automorphisms of $\alg{F}_0$ that leave $\alg{A}$ pointwise fixed. Finally, for a double cone $\mc{O}$, it is possible to define the local $^*$-subalgebra $\alg{F}_0(\mc{O})$ of $\alg{F}_0$, consisting of elements $(A, \rho, \psi)$, with $A \in \alg{A}(\mc{O})$, $\psi \in E(\rho)$, and $\rho$ localized in $\mc{O}$. 

To construct the field net, a faithful, $G$-invariant positive linear projection (in fact, a \emph{conditional expectation}) $m: \alg{F}_0 \to \alg{A}$ is defined. If $\omega_0$ is the vacuum state of $\alg{A}$, the GNS construction on the state $\omega_0 \circ m$ is used to create a representation $(\pi, \mc{H})$ of $\alg{F}_0$. The local algebras are then defined by $\alg{F}(\mc{O}) = \pi(\alg{F}_0(\mc{O}))''$. As usual, the algebra $\alg{F}$ is defined to be the norm closure of the union of all local algebras. Since $m$ is $G$-invariant, the action of $\alpha_g$ is implemented on $\mc{H}$ by unitaries $U(g)$. In other words, $\pi(\alpha_{g}(F)) = U(g) \pi(F) U(g)^*$ for $g\in G$ and $F \in \alg{F}_0$. This action can be extended to $\alg{F}$ in an obvious way. With these definitions, $(\pi, G, \alg{F})$ is a complete normal field net for $(\alg{A}, \omega_0)$ with local commutation relations. In fact, any complete normal field net for $\alg{A}$ is equivalent to the field net constructed here.

The final technical lemma concerns field operators. In the field net there are field operators, which can be interpreted as operators creating the DHR charges from the vacuum state. That is, for a DHR endomorphism $\rho$ there are $\Psi \in \alg{F}$ such that $\rho(A) \Psi = \Psi A$, with $A \in \alg{A}$. It is convenient in calculations to know how this works on the auxiliary algebras. 
\begin{lemma}
Let $\rho$ be an endomorphism of $\alg{A}$ localised in a double cone $\mc{O}$, and take $\psi \in E(\rho)$. Then 
\begin{equation}
  \pi^{\spc{S}_a}( \rho^{\spc{S}_a}(A)) \pi(I, \rho, \psi) = \pi(I, \rho, \psi) \pi^{\spc{S}_a}(A),
  \label{eq:vecmul}
\end{equation}
for all $A \in \alg{A}^{\spc{S}_a}$.
\label{lem:vecmul}
\end{lemma}
\begin{proof}
  Note that for $A \in \alg{A}$, the equality holds basically by construction of the field net. Now suppose $A \in \alg{A}( (\spc{S}_a+x)')''$. Then there is a net (in the sense of topology) $A_\lambda \to A$ in $\alg{A}( (\spc{S}_a+x)')$ that converges weakly to $A$. Equation~\eqref{eq:vecmul} holds for $A_\lambda$ by the previous remark. The result now follows by weak continuity of the extensions and of separate weak continuity of multiplication.
\end{proof}

\section{Extension to the field net}
\label{sec:extend}
Our next goal is to understand the BF-superselection structure of $\alg{F}$, including the way it is related to that of $\alg{A}$. Now that we have established how the auxiliary algebra is included in the field net, a natural question is how BF representations of $\alg{A}$ can be extended to BF representations of $\alg{F}$. This section is devoted to this problem. At the end of the section we comment on alternative approaches.

If $\widehat{\eta} \in \bfm^{\alg{F}}(\spc{C})$ is an extension of $\eta \in \bfm^{\alg{A}}(\spc{C})$, it follows that 
\[
\alpha_g \circ \widehat{\eta}(A) = \alpha_g \circ \eta(A) = \eta(A) = \widehat{\eta}\circ \alpha_g(A)
\]
for all $A \in \alg{A}$. The next theorem gives a characterisation of extensions such that $\alpha_g \circ \widehat{\eta}(F) = \widehat{\eta}\circ \alpha_g(F)$ for all $F \in \pi(\alg{F}_0)$. Such extensions are in 1-1 correspondence with certain families of unitaries $W_\rho(\eta)$ in $\alg{A}^{\spc{S}_a}$. A proof of this result for extensions of automorphisms was given in~\cite[Thm. 8.2]{MR1005608}. Later, the result of Doplicher and Roberts was adapted to endomorphisms~\cite{MR1721563}. The explicit description of the field net allows us to verify this construction, without invoking e.g.~universality properties as in the original proof.

The first step is to show that we can define an extension on the subalgebra $\pi(\alg{F}_0)$ of $\alg{F}$. We will then extend this to the algebra $\alg{F}$.
\begin{prop}
  Let $\eta$ be a representation of $\alg{A}$. Then representations $\widehat{\eta}$ of $\pi(\alg{F}_0)$ that extend $\eta$ and commute with $\alpha_g$ are in one-to-one correspondence with mappings $(\rho,\eta) \mapsto W_\rho(\eta)$ from $\dhr^{\alg{A}} \times \bfm^{\alg{A}}(\spc{C})$ to unitaries in $\alg{A}^{\spc{S}_a}$ satisfying
  \begin{align}
	W_\rho(\eta) &\in \Hom_\alg{A}(\rho \otimes \eta, \eta \otimes \rho), \label{eq:exthom}; \\
	W_{\rho'}(\eta)(T \otimes I_\eta) &= (I_\eta \otimes T) W_{\rho}(\eta), \quad T \in \Hom_\alg{A}(\rho,\rho'),  \label{eq:extnat}; \\
	W_{\rho \otimes \rho'}(\eta) &= (W_\rho(\eta) \otimes I_{\rho'})(I_\rho \otimes W_{\rho'}(\eta)), \label{eq:extnatrho};\\
	W_\rho(\eta \otimes \eta') &= (I_\eta \otimes W_\rho(\eta'))(W_\rho(\eta) \otimes I_{\eta'}) \label{eq:extnateta}.
  \end{align}
The extension is determined by
\begin{equation}
  \label{eq:extdefine}
  \widehat{\eta}(\pi(A, \rho, \psi)) = \pi^{\spc{S}_a}( \eta^{\spc{S}_a}(A) W_\rho(\eta))\pi(I, \rho, \psi).
\end{equation}
Moreover, if $S \in \Hom_{\alg{A}}(\eta, \eta')$ satisfies $S W_{\rho}(\eta) = W_\rho(\eta') \rho^{\spc{S}_a}(S)$ for all $\rho \in \dhr^{\alg{A}}$ (that is, $W_\rho(\eta)$ is natural in $\eta$), then $\pi^{\spc{S}_a}(S) \in \Hom_{\alg{F}_0}(\hat{\eta}, \hat{\eta}')$.
  \label{prop:extend}
\end{prop}
\begin{proof} To avoid cumbersome notation, $\pi^{\spc{S}_a}(\alg{A}^{\spc{S}_a})$ will be identified with $\alg{A}^{\spc{S}_a}$ in the proof.
  First, assume $\widehat{\eta}$ is a representation of $\alg{F}$ that commutes with the $G$-action. Lemma~\ref{lem:fieldfix} implies that $\widehat{\eta}$ restricts to a representation of $\alg{A}^{\spc{S}_a}$, which we will denote by $\eta$. For $\rho \in \dhr^{\alg{A}}$, write $\Psi_i = \pi(I, \rho, \psi_i)$, where $\psi_i$ is an orthonormal basis of $E(\rho)$. Define
\[
  W_\rho(\eta) = \sum_{i=1}^d \widehat{\eta}(\Psi_i) \Psi_i^*.
\]
This definition is independent of the chosen basis of $E(\rho)$. The $\Psi_i$ generate a Hilbert space with support $I$,~\cite[Proposition 270]{halvapp}, from which it follows that $W_\rho(\eta)$ is unitary. The Hilbert space $E(\rho)$ transforms according to some irreducible representation. Since $\widehat{\eta}$ commutes with the $G$-action, it is easy to verify that $\alpha_g(W_\rho(\eta)) = W_\rho(\eta)$. By Lemma~\ref{lem:fieldfix}(ii), $W_\rho(\eta)$ is a unitary in $\alg{A}^{\spc{S}_a}$. Note that $W_\iota(\eta) = I$, since $\eta$ is unital. Also note that for $\psi \in E(\rho)$, it follows that $W_\rho(\eta) \pi(I, \rho, \psi) = \widehat{\eta}(\pi(I,\rho,\psi))$. Because~\eqref{eq:extdefine} is in particular a $*$-endomorphism (see below for a verification) and $\alg{F}_0$ is generated by elements of this form, we see that $\widehat{\eta}$ can indeed be defined as in~\eqref{eq:extdefine}.

It remains to verify properties~\eqref{eq:exthom}--\eqref{eq:extnateta}. The verification of these properties is quite straightforward. We give a proof of~\eqref{eq:extnat} and leave the rest to the reader. So, let $T \in \Hom_\alg{A}(\rho, \rho')$. Note that $T \in \alg{A}$ by Haag duality for double cones. Then
\begin{align*}
	\sum_i \widehat{\eta}(\pi(T,\rho,\psi_i)) \pi(I, \rho, \psi_i)^* &= \sum_i \widehat{\eta}(\pi(I, \rho', E(T) \psi_i)) \pi(I, \rho, \psi_i)^* \\
  &= \sum_i \pi^{\spc{S}_a}(W_{\rho'}(\eta)) \pi(I, \rho', E(T) \psi_i) \pi(I, \rho, \psi_i)^* \\
  &= \pi^{\spc{S}_a}(W_{\rho'}(\eta)) \pi(T, \iota, 1).
\end{align*}
This is equation~\eqref{eq:extnat}. In the second line equation~\eqref{eq:extdefine} has been used.

As for the converse, we have to show that equation~\eqref{eq:extdefine} indeed defines a $*$-representation of $\pi(\alg{F}_0)$ that extends $\eta$. For $(A, \rho, \psi) \in \alg{F}_0$, define $\hat{\eta}( \pi (A, \rho, \psi) )$ as in equation~\eqref{eq:extdefine}. Note that~\eqref{eq:extnatrho} together with the unitarity of $W_{\iota}(\eta)$ imply that $W_{\iota}(\eta) = I$. Considering the embedding of $\alg{A}$ into $\alg{F}_0$ (by $A \mapsto (A, \iota, 1)$), it follows that $\hat{\eta}(\pi(A, \iota, 1)) = \pi^{\spc{S}_a}(\eta(A))$. This shows that we can view $\hat{\eta}$ as an extension of $\eta$.
  
To check that $\widehat{\eta}$ is well-defined, suppose $(AT, \rho, \psi) = (A, \rho', E(T) \psi)$, with $T$ intertwining $\rho$ and $\rho'$. A simple computation, using $\pi^{\spc{S}_a}(T) = \pi(T)$, and the fact that $\pi$ is well-defined, shows that well-definedness of $\hat{\eta}$ boils down to the identity
\[
\eta(A) W_{\rho'}(\eta) T = \eta(AT) W_{\rho}(\eta),
\]
which in turn is easily verified using the properties of $W_\rho(\eta)$.

In order to show that $\widehat{\eta}$ is multiplicative, consider $F = (A, \rho, \psi)$ and $F' = (A', \rho', \psi')$ as elements of $\alg{F}_0$. Then:
\begin{equation}
  \begin{split}
	\hat{\eta}(\pi(F)\pi(F')) &= \hat{\eta}(\pi(A \rho(A'), \rho \otimes \rho', \psi \otimes \psi')) \\
	&= \pi^{\spc{S}_a}(\eta(A \rho(A')) W_{\rho\otimes\rho'}(\eta) \pi(I, \rho \otimes \rho', \psi \otimes \psi').
  \end{split}
  \label{eq:extmult}
\end{equation}
On the other hand,
\[
\hat{\eta}(\pi(F)) \hat{\eta}(\pi(F')) = \pi^{\spc{S}_a}(\eta(A)W_{\rho}(\eta))\pi(I, \rho, \psi) \pi^{\spc{S}_a} (\eta(A') W_{\rho'}(\eta)) \pi(I, \rho', \psi').
\]
An application of Lemma~\ref{lem:vecmul} reduces the right hand side to
\[
\pi^{\spc{S}_a}(\eta(A) W_{\rho}(\eta) \rho^{\spc{S}_a}(\eta(A') W_{\rho'}(\eta))) \pi(I, \rho \otimes \rho', \psi \otimes \psi').
\]
Then one should note that $W_{\rho}(\eta)$ intertwines $\rho^{\spc{S}_a}\circ \eta$ and $\eta^{\spc{S}_a} \circ \rho$, and use the fact that $\rho$ is an endomorphism of $\alg{A}$, so that $\eta^{\spc{S}_a}(\rho(A')) = \eta(\rho(A'))$. By using~\eqref{eq:extnatrho}, one then obtains equation~\eqref{eq:extmult}, so $\widehat{\eta}$ preserves multiplication.

To check that $\hat{\eta}$ is a $^*$-homomorphism, we have to show $\widehat{\eta}(\pi(F)^*) = \widehat{\eta}(\pi(F))^*$. Since $\hat{\eta}$ preserves multiplication, it is enough to show this for $(A, \iota, 1)$ and $(I, \rho, \psi) \in \alg{F}_0$. The first case is easy:
\[
\hat{\eta}(\pi(A, \iota, 1)^*) = \hat{\eta}(\pi(A^*, \iota, 1)) = \pi^{\spc{S}_a}(\eta(A^*)) \pi(I, \iota, 1) = \pi^{\spc{S}_a}(\eta(A))^*,
\]
since $\eta$ and $\pi^{\spc{S}_a}$ are $^*$-homomorphisms. To check the remaining case, let $(\overline{\rho}, R, \overline{R})$ be a conjugate. Then, $R^* \in \Hom_{\alg{A}}(\overline{\rho}\otimes\rho, \iota)$, so we have
\begin{equation}
	\begin{split}
  \eta(R^*) W_{\overline{\rho}}(\eta) = &W_{\iota}(\eta) R^* W_{\overline{\rho} \otimes \rho}(\eta)^* W_{\overline{\rho}}(\eta) = R^* (W_{\overline{\rho}}(\eta)^* W_{\overline{\rho} \otimes \rho}(\eta))^* \\ &= R^* \overline{\rho}^{\spc{S}_a}(W_{\rho}(\eta)^*),
	\end{split}
  \label{eq:extconj}
\end{equation}
where the properties of $W_\rho(\eta)$ have been used in each step. Recall the anti-linear map $\mc{J}$ used in the definition of the $*$-operation on $\alg{F}_0$. Then, by definition of $\widehat{\eta}$,
\[
\begin{split}
	\hat{\eta}(\pi(I, \rho, \psi)^*) = \hat{\eta}(\pi(R^*, \overline{\rho}, (&\mc{J}E(\overline{R}^*))\psi)) \\
	= \pi^{\spc{S}_a}(&\eta(R^*) W_{\overline{\rho}}(\eta)) \pi(I, \overline{\rho}, (\mc{J}E(\overline{R}^*)) \psi).
\end{split}
\]
Substitute equation~\eqref{eq:extconj} and apply Lemma~\ref{lem:vecmul}. Together with the fact that $\pi^{\spc{S}_a}$ agrees with $\pi$ on $\alg{A}$, this gives
\begin{align*}
  \widehat{\eta}(\pi(I, \rho, \psi)^*) &= \pi^{\spc{S}_a}(R^* \overline{\rho}^{\spc{S_a}}(W_{\rho}(\eta)^*) \pi(I, \overline{\rho}, (\mc{J}E(\overline{R}^*)) \psi)  \\
  &= \pi(R^*, \iota, 1) \pi(I, \overline{\rho}, (\mc{J}E(\overline{R}^*)) \psi) \pi^{\spc{S}_a}(W_{\rho}(\eta)^*) \\
  &= \pi(I, \rho, \psi)^* \pi^{\spc{S}_a}(W_{\rho}(\eta))^* \\
  &= \widehat{\eta}(\pi(1, \rho, \psi))^*,
\end{align*}
which concludes the proof that $\widehat{\eta}$ is a representation. 

To prove that $\widehat{\eta}$ commutes with the $G$-action, consider $(A, \rho, \psi) \in \alg{F}_0$, and let $g \in G$. Then
\[
\widehat{\eta}(\alpha_g \pi(A, \rho, \psi)) = \widehat{\eta}(\pi(A, \rho, g_\rho \psi)) = \pi^{\spc{S}_a}(\eta(A) W_\rho(\eta) \pi(I, \rho, g_\rho \psi).
\]
On the other hand, $\alpha_g$ is implemented by $U(g)$, so we have
\begin{align*}
  \alpha_g \circ \widehat{\eta}(\pi(A,\rho,g_\rho)) &= U(g) \pi^{\spc{S}_a}(\eta(A) W_{\rho}(\eta)) \pi(I, \rho, \psi) U(g)^* \\
  &= U(g) \pi^{\spc{S}_a}(\eta(A) W_{\rho}(\eta)) U(g)^* \pi(I,\rho,g_\rho\psi).
\end{align*}
From this it follows that if $\pi^{\spc{S}_a}(\eta(A) W_{\rho}(\eta))$ is $G$-invariant, then $\widehat{\eta}$ commutes with the action of $G$. Since $\eta(A) W_{\rho}(\eta) \in \alg{A}^{\spc{S}_a}$ this is nothing but Lemma~\ref{lem:fieldfix}(ii).

Finally, let $S \in \Hom_\alg{A}(\eta, \eta')$ be an intertwiner, and $F = (A, \rho, \psi) \in \alg{F}_0$. Then
\begin{align*}
  \pi^{\spc{S}_a}(S) \widehat{\eta}(\pi(F)) &= \pi^{\spc{S}_a} (S \eta(A) W_{\rho}(\eta)) \pi(1, \rho \psi) \\
  &= \pi^{\spc{S}_a}(\eta'(A) S W_{\rho}(\eta)) \pi(I, \rho, \psi) \\
  &= \pi^{\spc{S}_a}(\eta'(A) W_{\rho}(\eta') \rho^{\spc{S}_a}(S)) \pi(I, \rho, \psi) \\
  &= \widehat{\eta}'(\pi(F)) \pi^{\spc{S}_a}(S),
\end{align*}
where in the last line Lemma~\ref{lem:vecmul} has been used. Hence we see that $\pi^{\spc{S}_a}(S) \in \Hom_{\alg{F}_0}(\widehat{\eta}, \widehat{\eta}')$, completing the proof.
\end{proof}
It should be noted that conditions~\eqref{eq:exthom}--\eqref{eq:extnateta} are very similar to the conditions on a braiding, in particular the braiding $\varepsilon_{\rho,\eta}$ satisfies these conditions. The only difference is that $W_{\rho}(\eta)$ need only be defined for $\rho$ a DHR endomorphism and $\eta$ a BF endomorphism. 

The construction above gives an extension of representations of $\alg{A}$ to $\alg{F}$. To verify if these extensions are BF representations one should look at the localisation properties of the extension. The next lemma gives a necessary and sufficient condition for the extension of a localised representation to be cone localised again.
\begin{lemma}
  \label{lem:extlocal}
  Consider the notation and assumptions of Proposition~\ref{prop:extend}. If $\eta$ is localised in $\spc{C}$, its extension $\widehat{\eta}$ is localised in $\spc{C}$ if and only if $W_\rho(\eta) = I$ for each $\rho \in \dhr^{\alg{A}}$ localised spacelike to $\spc{C}$. Here, $\widehat{\eta}$ is called localised in $\spc{C}$ if it acts trivially on all $F \in \pi(\alg{F}_0(\mc{O}))$ for $\mc{O} \subset \spc{C}'$.
\end{lemma}
\begin{proof}
  The localisation properties follow from the localisation of $\eta$. If $F \in \alg{F}_0(\mc{O})$ for some double cone $\mc{O} \subset \spc{C}'$, it is of the form $F = (A, \rho, \psi)$, with $A \in \alg{A}(\mc{O})$ and $\rho$ localised in $\mc{O}$. But $\eta$ acts trivially on such $A$, and $W_{\rho}(\eta) = I$. Hence $\widehat{\eta}(\pi(A, \rho, \psi)) = \pi(A, \rho, \psi)$.

  For the converse, suppose that $\rho \in \dhr^{\alg{A}}$ is localised spacelike to $\spc{C}$. Choose an orthonormal basis $\psi_i$ of $E(\rho)$. Then $\pi(I, \rho, \psi_i) \in \pi(\alg{F}_0(\mc{O}))$ for $\mc{O} \subset \spc{C}'$. Hence
\[
\widehat{\eta}(\pi(I, \rho, \psi_i)) = \pi^{\spc{S}_a}(W_\rho(\eta)) \pi(I, \rho, \psi_i) = \pi(I, \rho, \psi_i).
\]
We multiply on the right by $\pi(I, \rho, \psi_i)^*$ and sum over $i$. Since $E(\rho)$ has support $I$, it follows that $\pi^{\spc{S}_a}(W_\rho(\eta))$ is the identity. 
\end{proof}

As a consequence of these results, we can canonically extend BF representations of $\alg{A}$ to BF representations of $\alg{F}$. This way of extending representations was first pointed out by Rehren~\cite{MR1128146}, where the author sketches a proof in the case of compactly localised sectors.
\begin{theorem}
  \label{thm:unique}
  Every BF representation $\eta$ of $\alg{A}$ can be extended to a BF representation of $\alg{F}$ that commutes with the $G$-action. This extension is unique.
\end{theorem}
\begin{proof}
 One readily verifies that $W_\rho(\eta) = \varepsilon_{\rho,\eta}$ has the properties required in Proposition~\ref{prop:extend}. Moreover, $W_\rho(\eta) = I$ if $\rho$ is localised spacelike to $\eta$. Hence there is a $^*$-representation $\widehat{\eta}$ of $\pi(\alg{F}_0)$ extending $\eta$. If $\eta$ is localised in $\spc{C}$, Lemma~\ref{lem:extlocal} shows that $\widehat{\eta}$ is localised in the same region. If $\widetilde{\spc{C}}$ is another spacelike cone, by transportability of $\eta$ there is a unitarily equivalent $\eta'$ localised in $\widetilde{\spc{C}}$. By Proposition~\ref{prop:extend}, this lifts to a unitary equivalence of $\widehat{\eta}$ and $\widehat{\eta}'$, since the condition stated on $S$ is nothing but naturality of $\varepsilon_{\rho,\eta}$ in $\eta$. This shows transportability of the extension.

 We now have a representation defined on the algebra $\pi(\alg{F}_0)$. To extend this representation to $\alg{F}$, we first show that it can be extended to the local algebras $\alg{F}(\mc{O}) = \pi(\alg{F}_0(\mc{O}))''$. Consider a double cone $\mc{O}$. If $\mc{O}$ is spacelike to $\spc{C}$, localisation implies $\widehat{\eta}(\pi(F)) = \pi(F)$ for all $\pi(F) \in \pi(\alg{F}_0(\mc{O}))$. In this case it is clear that this extends to the weak closure $\alg{F}(\mc{O})$. Now suppose $\mc{O}$ is not spacelike to $\spc{C}$. Then by the argument above, there is a unitary $V$ such that $\widetilde{\eta}(\pi(F)) = V^* \widehat{\eta}(\pi(F)) V$ which is localised spacelike to $\mc{O}$. In other words, $\widehat{\eta}(\pi(F)) = V \pi(F) V^*$, by localisation of $\widetilde{\eta}$. The right hand side is weakly continuous, hence we can extend $\widehat{\eta}$ to $\alg{F}(\mc{O})$ for every $\mc{O}$. But the argument also shows that $\widehat{\eta}$ is in fact an isometry, since $\| V \pi(F) V^* \| = \| \pi(F) \|$. The union of the local algebras is norm dense in $\alg{F}$, hence by continuity $\widehat{\eta}$ extends uniquely to a representation of $\alg{F}$.
 
  Finally, we show that the extension is unique. Suppose that we have another localised extension that commutes with the action of $G$. Proposition~\ref{prop:extend} then asserts the existence of a family $W_\rho(\eta)$. We show $W_\rho(\eta) = \varepsilon_{\rho,\eta}$. First of all, suppose $\rho \in \dhr^{\alg{A}}$ is localised spacelike to the localisation of $\eta$. Then, by Lemma~\ref{lem:extlocal}, $W_\rho(\eta) = I$. But this is equal to $\varepsilon_{\rho,\eta}$, since $\rho$ is degenerate. Now consider an arbitrary $\rho \in \dhr^{\alg{A}}$. Choose a unitary equivalent $\rho'$ localised spacelike to the localisation of $\eta$, with corresponding unitary $T$. Then,
\[
	(T \otimes I_\eta) = (I_\eta \otimes T)W_\rho(\eta), \quad (T \otimes I_\eta) = (I_\eta \otimes T) \varepsilon_{\rho,\eta},
\]
where the first equation follows from~\eqref{eq:extnat}, and the second follows from naturality with respect to $\rho$ of the braiding. Since $T$ is a unitary, it follows that $W_\rho(\eta) = \varepsilon_{\rho,\eta}$.
\end{proof}
\begin{remark} (i) Localisation properties are used to show that $\widehat{\eta}$ can be extended to a representation of $\alg{F}$. By applying the results of~\cite{MR1005608}, as in~\cite{MR1721563}, it can be proved that in fact \emph{every} extension (whether it is cone localised or not) as in Proposition~\ref{prop:extend} can be defined on the whole of $\alg{F}$.

	(ii) Denote the canonical extension by $\Phi(\eta)$ or $\widehat{\eta}$. It turns out that $\Phi: \eta \mapsto \widehat{\eta}$ is in fact a faithful, but not full, tensor functor. These and other categorical aspects are discussed in Section~\ref{sec:crossed}.
\end{remark}

Let us briefly comment on other approaches to the problem of extending representations. Firstly one could use techniques from the theory of subfactors. For this to work $\alg{A}(\spc{C})'' \subset \alg{F}(\spc{C})''$ needs to be an inclusion of factors. Moreover, the Jones index of this inclusion should be finite. In this case the machinery of $\alpha$-induction and $\sigma$-restriction can be applied~\cite{MR1652746}. In the present situation, however, it is not clear if these requirements are satisfied.

Another approach that can be used in the DHR setting is Roberts' theory of localised cocycles~\cite{MR0471753,MR1147461}, see also~\cite{MR1828981}. It is not immediately clear, however, if this can be modified to apply to case of BF sectors. For one, the set of all double cones is directed, unlike the set of all spacelike cones.

\section{Non-abelian cohomology and restriction to the observable algebra}
\label{sec:restrict}
In the previous section, extension of BF representations of the observable algebra to the field algebra was discussed. Here we investigate the other direction: does every BF representation of the field algebra that commutes with the group action come from such an extension? This is a first step in understanding the category $\bfm^\alg{F}(\spc{C})$. In answering this question, one encounters problems of a cohomological nature in a natural way. 

For convenience of the reader we recall the notion of an $\alpha$-1-cocycle and an $\alpha$-2-cocycle in a von Neumann algebra $\alg{M}$; for the complete definition see~\cite{MR574031II}. A Borel map $v: G \to \mc{U}(\alg{M})$ is an \emph{$\alpha$-1-cocycle} if it satisfies the identity
\[
	v(gh) = \alpha_g(v(h)) v(g);
\]
a map $w: G \times G \to \mc{U}(\alg{M})$ is an \emph{$\alpha$-2-cocycle} if
\[
	\quad w(gh, k) w(g,h) = w(g,hk) \alpha_g(w(h,k)).
\]
It is possible to define a coboundary map $\partial$. For example, a 1-cocycle $v(g)$ is a coboundary if there is a unitary $w \in \alg{M}$ such that $v(g) = \alpha_g(w) w^*$. A 2-cocycle $w(g,h)$ is a coboundary if there is a Borel map $\psi: G \to \mc{U}(\alg{M})$ such that $w(g,h) = \alpha_g(\psi(h)) \psi(g) \psi(gh)^*$.

It turns out that each cocycle taking values in $\alg{F}(\spc{C})$ is in fact a coboundary in a bigger algebra $\alg{F}(\widetilde{\spc{C}})'' \supset \alg{F}(\spc{C})''$. This is essentially due to the field net having full $G$-spectrum, which allows to use the construction of Sutherland to construct a coboundary~\cite{MR574031II}. Before proving this result, we first recall some notions regarding Hilbert spaces in von Neumann algebras~\cite{MR0473859}.

\begin{definition}
  \label{def:hspace}
	Let $\alg{M}$ be a von Neumann algebra. A Hilbert space in $\alg{M}$ is a norm closed linear subspace $H$, such that $a \in H$ implies $a^*a \in \mathbb{C} I$ and $x \in \alg{M}$, $ax = 0$ for all $a \in H$ implies $x=0$. 
\end{definition}
An inner product is then defined by $(a,b)I = a^*b$. One can check that this indeed defines a Hilbert space. If $\{V_i\}_{i \in J}$ is an orthonormal basis for $H$, the operators $V_i V_i^*$ are (mutually orthogonal) projections, hence the $V_i$ are isometries, and $\sum_{i \in J} V_i V_i^* = I$.  Certain operators $x \in \alg{M}$ can be identified with operators in $\alg{B}(H)$. More generally, if $H_1$ and $H_2$ are two Hilbert spaces in $\alg{M}$, write
\[
(H_1, H_2) = \{ x \in \alg{M} : \psi_2^* x \psi_1 \in \mathbb{C} I, \psi_1 \in H_1, \psi_2 \in H_2 \}.
\]
These operators are in 1-1 correspondence with operators in $\alg{B}(H_1, H_2)$, see~\cite[Lemma 2.3]{MR0473859}. For $x \in (H_1, H_2)$, write $L(x)$ for the corresponding linear operator in $\alg{B}(H_1, H_2)$. In this case, $(\psi_1, L(x) \psi_2) I = \psi_1^* x \psi_2$. With these preparations we can prove the triviality of cocycles.
\begin{theorem}
  \label{thm:cobound}
	Assume $G$ is second countable. Let $v(g_1, \dots, g_n)$ be a unitary $\alpha$-$n$-cocycle in $\alg{F}(\spc{C})''$. Then there is a spacelike cone $\widetilde{\spc{C}} \supset \spc{C}$ such that $v$ is a coboundary in $\alg{F}(\widetilde{\spc{C}})''$.
\end{theorem}
\begin{proof} 
Pick a double cone $\mathcal{O} \subset \spc{C}'$, such that there is a spacelike cone $\widetilde{\spc{C}} \supset \spc{C} \cup \mathcal{O}$. Note that this is always possible. Since the field net has full spectrum, for each irreducible representation $\xi$ of $G$, there is a Hilbert space in $\alg{F}(\mc{O})$, transforming according to this representation. That is, there are isometries $\psi_i$, $i = 1, \dots, d$ spanning a Hilbert space $H_\xi$ in $\alg{F}(\mc{O})$, such that
\[
	\alpha_g(\psi_i) = \sum_{j=1}^d u^\xi_{ji}(g) \psi_j,
\]
where $u^\xi_{ji}(g)$ are the matrix coefficients of $\xi$.

The left regular action $\lambda(g)$ on $L^2(G)$ decomposes as a direct sum of irreducible representations. By the Peter-Weyl theorem the Hilbert space $L^2(G)$ decomposes as~\cite{MR0262773}
\begin{equation}
  \label{eq:peterweyl}
  L^2(G) = \bigoplus_{\xi \in \widehat{G}} d_\xi H_\xi,
\end{equation}
where $d_\xi$ is the dimension of the representation $\xi$. For each irreducible representation $\xi$, the algebra $\alg{F}(\mc{O})$ contains a Hilbert space $H_\xi$ (as in Definition~\ref{def:hspace}), transforming according to the corresponding representation. The group $G$ is second countable, hence the number of irreducible representations is at most countable~\cite{MR0262773}. Since $\alg{A}(\mc{O})$ is a properly infinite von Neumann algebra acting on a separable Hilbert space, it is possible to find a countable family of isometries $V_i$ such that $V_i^* V_j = \delta_{i,j} I$ and $\sum_i V_i V_i^* = I$. Moreover, they are invariant under the action of $G$. These isometries enable us to construct an image of the direct sum decomposition~\eqref{eq:peterweyl} of $L^2(G)$ in $\alg{F}(\mc{O})$ as follows. First choose an enumeration $\xi_i$ of $\widehat{G}$, counted with multiplicities. For each $i$ choose an orthonormal basis $\psi_j$ of $H_{\xi_i}$ where $j = 1, \dots, d_{\xi_i}$. Then $e_{ij} = V_i \psi_j V_i^*$ forms an orthonormal basis of a Hilbert space in $\alg{F}(\mc{O})$. This Hilbert space will be denoted by $L^2_{\alg{F}}(G)$. If $T: L^2_{\alg{F}}(G) \to L^2(G)$ denotes the corresponding isomorphism of Hilbert spaces, the above remarks imply that $T(\alpha_g(\psi)) = \lambda(g) T(\psi)$ for all $\psi \in L^2_{\alg{F}}(G)$.

Note that the action $\alpha_g$ induces an action on $\alg{B}(L^2_\alg{F}(G))$. To see what effect this has on the corresponding operators in $\alg{B}(L^2(G))$, consider the following calculation, where $\langle -,-\rangle$ is the inner product of $L^2(G)$, $x \in \alg{B}(L^2_\alg{F}(G))$, and $g \in G$:
\begin{align*}
\langle T(\psi_1), L(x) T(\psi_2) \rangle I &= \psi_1^* x \psi_2 \\
	&= \alpha_g(\psi_1^*) \alpha_g(x) \alpha_g(\psi_2) \\
	&= (\alpha_g(\psi_1), L(\alpha_g(x)) \alpha_g(\psi_2)) I  \\ 
	&=\langle \lambda(g) T(\psi_1), L(\alpha_g(x)) \lambda(g) T(\psi_2) \rangle I \\
	&=\langle T(\psi_1), \lambda(g)^* L(\alpha_g(x)) \lambda(g) T(\psi_2) \rangle I. 
\end{align*}
In other words, $L(\alpha_g(x)) = \lambda(g) L(x) \lambda(g)^* = \Ad \lambda(g) L(x)$, since the left regular representation is unitary.

The situation can be summarised as follows: there is a copy of $L^2(G)$ in $\alg{F}(\mc{O})$, as well as a copy of $\mc{B}(L^2(G))$. Moreover, the action $\alpha_g$ of $G$ acts as $\Ad \lambda(g)$ on these operators. We are now in a position to apply Proposition 2.5.1 from~\cite{MR574031II}.

Define an injective representation $\pi: \alg{F}(\spc{C})'' \otimes \mc{B}(L^2(G)) \to \alg{F}(\widetilde{\spc{C}})''$ by $\pi(x \otimes y) = xF^{-1}(y)$. Note that this is indeed a representation, since $\alg{F}(\spc{C})''$ commutes with $\alg{F}(\mc{O})$. Endow the algebra $\alg{F}(\spc{C})'' \otimes \mc{B}(L^2(G))$ with the action $\beta_g$ of $G$ defined by $\beta_g = \alpha_g \otimes \Ad \lambda(g)$. It follows that for each $g \in G$, $\pi(\beta_g(x \otimes y)) = \alpha_g(\pi(x \otimes y))$. By Proposition 2.1.5 of~\cite{MR574031II} $v(g_1, \dots g_n) \otimes I$ is a $\beta$-coboundary. But since $v(g_1, \dots g_n) = \pi(v(g_1, \dots g_n) \otimes I)$ and $\alpha_g \circ \pi = \pi \circ \beta_g$, it follows that $v(g_1, \dots g_n)$ is an $\alpha$-coboundary in $\alg{F}(\widetilde{\spc{C}})''$.
\end{proof}
\begin{remark}The DHR sectors of $\alg{A}$ are in one-to-one correspondence with irreducible representations of the group $G$. Hence under the assumption already made in Theorem~\ref{thm:trivdhr}, it follows that $G$ is indeed second countable.
\end{remark}

With this theorem we are able to prove the main result of this section, namely that every BF representation of $\alg{F}$ that commutes with the $G$-action comes from the extension of a representation of $\alg{A}$.
\begin{corollary}
  \label{cor:restrict}
  Let $\eta \in \bfm^\alg{F}(\spc{C})$, such that $\alpha_g \circ \eta = \eta \circ \alpha_g$ for all $g \in G$. Then $\eta$ restricts to a BF sector $\eta \upharpoonright \alg{A}^{\spc{S}_a}$ of the observable net. Moreover, $\widehat{\eta \upharpoonright \alg{A}^{\spc{S}_a}} = \eta$.
\end{corollary}
\begin{proof}
  Since the representation $\eta$ commutes with the action of $G$, by Lemma~\ref{lem:fieldfix}(ii) it restricts to an endomorphism of $\alg{A}^{\spc{S}_a}$. It is clear that this restriction is localised in $\spc{C}$ as well. To prove transportability, proceed in a similar way as in~\cite[Proposition 3.5]{MR2183964}. Suppose $\widehat{\spc{C}}$ is another spacelike cone. For simplicity we assume it is spacelike to $\spc{S}_a$. In the general case, one has to apply an argument as in the proof of Proposition~\ref{prop:equiv}. Pick a spacelike cone $\widetilde{\spc{C}} \subset \widehat{\spc{C}}$ such that there is a double cone $\widehat{\spc{C}} \supset \mathcal{O} \subset \widetilde{\spc{C}}'$. By Lemma~\ref{lem:chargetr} and transportability, there is a unitary $V \in \alg{F}^{\spc{S}_a}$ such that $\widetilde{\eta} = \Ad V \circ \eta$ is localised in $\widetilde{\spc{C}}$.

  Now consider ${}^g \widetilde{\eta} = \alpha_g \circ \widetilde{\eta} \circ \alpha_{g^{-1}}$. Since $\eta$ is $G$-invariant, $\alpha_g(V) \in \Hom_{\alg{F}} (\eta, {}^g\widetilde{\eta})$. Because $\alpha_g$ leaves $\alg{F}(\widetilde{\spc{C}}')$ globally invariant, ${}^g \widetilde{\eta}$ is also localised in $\widetilde{\spc{C}}$. Define an $\alpha$-1-cocycle $v(g) = \alpha_g(V)V^* \in \Hom_\alg{F}(\widetilde{\eta}, {}^g\widetilde{\eta})$. By Haag duality, $v(g) \in \alg{F}(\widetilde{\spc{C}})''$. Moreover $g \mapsto v(g)$ is strongly continuous. By Theorem~\ref{thm:cobound} there is a unitary $W \in \alg{F}(\widehat{\spc{C}})''$ such that $v(g) = \alpha_g(W) W^*$. Define $\widehat{\eta} = \Ad W^* \circ \widetilde{\eta}$. It is easy to see that $\widehat{\eta}$ is localised in $\widehat{\spc{C}}$ and that $W^*V \in \Hom_\alg{F}(\eta, \widehat{\eta})$. Moreover, by definition $\alpha_g(V) V^* = \alpha_g(W) W^*$, from which it follows that $\alpha_g(W^*V) = W^*V$ for all $g \in G$. Hence $W^*V$ is in $\alg{A}^{\spc{S}_a}$, and is the desired intertwiner from $\eta \upharpoonright \alg{A}^{\spc{S}_a}$ to $\widehat{\eta} \upharpoonright \alg{A}^{\spc{S}_a}$.

  Since extensions commuting with $G$ are unique by~Theorem~\ref{thm:unique}, the last statement is obvious.
\end{proof}

\section{Categorical crossed products}
\label{sec:crossed}
The results in the previous section give a complete understanding of all $G$-invariant BF representations of $\bfm^{\alg{F}}(\spc{C})$. Indeed, these are all of the form $\Phi(\eta)$ for some BF representation $\eta$ of $\alg{A}$. Recall that this extension functor is defined by $\Phi(\eta) = \widehat{\eta}$, and by $\Phi(S) = \pi^{\spc{S}_a}(S)$ for intertwiners $S$ (see Proposition~\ref{prop:extend}). In fact, this extension preserves all relevant properties of the category $\bfm^{\alg{A}}(\spc{C})$.
\begin{prop}
	The functor $\Phi:\bfm^{\alg{A}}(\spc{C}) \to \bfm^{\alg{F}}(\spc{C})$ is a strict braided monoidal functor. It also preserves direct sums: $\Phi(\eta_1 \oplus \eta_2) \cong \Phi(\eta_1) \oplus \Phi(\eta_2)$. Finally, $d(\Phi(\eta)) = d(\eta)$.
\end{prop}
\begin{proof}
  Functoriality of $\Phi$ is immediate. Note that $\Phi(\iota)$ is just the identity endomorphism of $\alg{F}$, hence it preserves the tensor unit. We verify $\Phi(\eta_1 \otimes \eta_2) = \Phi(\eta_1) \otimes \Phi(\eta_2)$ on a dense subalgebra. Consider $F = (A, \rho, \psi) \in \alg{F}_0$. Then the extension of the tensor product is given by
  \begin{equation}
	  \widehat{\eta_1 \otimes \eta_2}(\pi(F)) = \pi^{\spc{S}_a}(\eta_1^{\spc{S}_a}\eta_2(A) \varepsilon_{\rho,\eta_1 \otimes \eta_2}) \pi(1, \rho, \psi).
	  \label{eq:tensprod}
  \end{equation}
Note that by definition, $\widehat{\eta_1}(\pi(A,\iota, 1)) = \pi^{\spc{S}_a}(\eta_1(A))$ for all $A \in \alg{A}$. Passing to the unique weakly continuous extension, and taking weak limits, it follows that $\widehat{\eta}^{\spc{S}_a}_1(\pi^{\spc{S}_a}(A)) = \pi^{\spc{S}_a}(\eta_1^{\spc{S}_a}(A))$ for all $A \in \alg{A}^{\spc{S}_a}$. We then calculate
\begin{align*}
  (\widehat{\eta_1} \otimes \widehat{\eta_2})(\pi(F)) &= \widehat{\eta_1}^{\spc{S}_a}(\pi^{\spc{S}_a}(\eta_2(A) \varepsilon_{\rho,\eta_2}) \pi(I,\rho,\psi)) \\
  &= \widehat{\eta_1}^{\spc{S}_a}(\pi^{\spc{S}_a}(\eta_2(A) \varepsilon_{\rho,\eta_2})) \pi^{\spc{S}_a}(\varepsilon_{\rho,\eta_1})\pi(I,\rho,\psi) \\
  &= \pi^{\spc{S}_a}(\eta_1^{\spc{S}_a}(\eta_2(A) \varepsilon_{\rho,\eta_2}) \varepsilon_{\rho,\eta_1}) \pi(I,\rho,\psi).
\end{align*} 
By the braid equations (cf. conditions~\eqref{eq:extnat}--\eqref{eq:extnateta}), the last line is equal to equation~\eqref{eq:tensprod}. For $\eta_1, \eta_2 \in \bfm^{\alg{A}}(\spc{C})$, note that $\Phi(\varepsilon_{\eta_1, \eta_2}) = \varepsilon_{\Phi(\eta_1), \Phi(\eta_2)}$. This follows from uniqueness of the braiding of $\bfm^{\alg{F}}(\spc{C})$, and by noticing that the funtor $\Phi$ sends spectator morphisms used in the definition of $\varepsilon_{\eta_1, \eta_2}$ to spectator morphisms for $\Phi(\eta_1)$ and $\Phi(\eta_2)$.

To prove that $\Phi$ preserves direct sums, assume $\eta_1 \oplus \eta_2 = \Ad V_1\circ \eta_1 + \Ad V_2 \circ \eta_2$. It is then not hard to show that for $F \in \alg{F}_0$,
\[
	\Phi(\eta_1 \oplus \eta_2)(\pi(F)) = \Phi(V_1) \Phi(\eta_1)(\pi(F)) \Phi(V_1^*) + \Phi(V_2) \Phi(\eta_1)(\pi(F)) \Phi(V_2^*).
\]
The right hand side is just the direct sum $\Phi(\eta_1) \oplus \Phi(\eta_2)$.

Finally, for the last statement one can show that if $(\overline{\eta}, R, \overline{R})$ is a standard conjugate for $\eta$, then $(\Phi(\overline{\eta}), \Phi(R), \Phi(\overline{R}))$ is a standard conjugate for $\Phi(\eta)$, and this determines the dimension. Details can be found in~\cite[Proposition 344]{mmappendix}.
\end{proof}

Using some harmonic analysis, the intertwiners between two extensions can be described explicitly.
\begin{prop}
	For $\gamma \in \dhr^{\alg{A}}$, write $H_\gamma$ for the Hilbert space in $\alg{F}$ generated by $\pi(I, \gamma, \psi), \psi \in E(\gamma)$. Then for $\eta_1, \eta_2 \in \bfm^{\alg{A}}(\spc{C})$,
	\begin{equation}
	\label{eq:homsets}
	  \Hom_{\alg{F}}(\Phi(\eta_1), \Phi(\eta_2)) = \lspan_{i \in \widehat{G} } \pi^{\spc{S}_a}(\Hom_\alg{A}(\gamma_i \otimes \eta_1, \eta_2))H_{\gamma_i},
	\end{equation}
where $\gamma_i \in \dhr^{\alg{A}}$ corresponds to the irrep $i$. Moreover, we can choose each $\gamma_i$ to be localised in a double cone $\mc{O}_i \subset \spc{C}$.
 \label{prop:homsets}
\end{prop}
\begin{proof}
	Consider $T \in \Hom_{\alg{A}}(\gamma \otimes \eta_1,\eta_2)$ and $\Psi = \pi(I, \gamma, \psi) \in H_\gamma$. By Proposition~\ref{prop:extend}, $T$ lifts to an intertwiner $\pi^{\spc{S}_a}(T)$ from $\widehat{\gamma\otimes\eta_1}$ to $\widehat{\eta_2}$, hence
\[
 \widehat{\eta_2}(\pi(A, \rho, \psi')) \pi^{\spc{S}_a}(T) \Psi = \pi^{\spc{S}_a}(T) \widehat{\gamma\otimes\eta_1}(\pi(A,\rho, \psi')) \Psi.
\]
Since the DHR morphisms form a symmetric category and $E$ is a symmetric $^*$-tensor functor, that is, it maps $\varepsilon_{\gamma,\rho}$ to the canonical symmetry $\Sigma_{E(\gamma), E(\rho)}$, it follows that $\pi(I, \rho, \psi') \pi(I, \gamma, \psi) = \pi(\varepsilon_{\gamma,\rho}, \gamma, \psi) \pi(I, \rho, \psi')$. Using the braid equations, we then have
\begin{align*}
  \pi^{\spc{S}_a}(\gamma^{\spc{S}_a}\eta_1(A) \varepsilon_{\rho,\gamma\otimes\eta_1}) \pi(I, \rho, \psi') \Psi &= \pi^{\spc{S}_a}(\gamma^{\spc{S}_a}\eta_1(A) \varepsilon_{\rho,\gamma \otimes \eta_1} \varepsilon_{\gamma,\rho}) \Psi \pi(I, \rho, \psi') \\
  &= \pi^{\spc{S}_a}(\gamma^{\spc{S}_a}(\eta_1(A) \varepsilon_{\rho,\eta_1})) \Psi \pi(I, \rho, \psi').
\end{align*}
An application of Lemma~\ref{lem:vecmul} then shows that $\pi^{\spc{S}_a}(T) \Psi \in \Hom_{\alg{F}}(\Phi(\eta_1), \Phi(\eta_2))$.

For the other direction, note that since $\Phi(\eta_1)$ and $\Phi(\eta_2)$ are $G$-invariant extensions, it follows that $\Hom_{\alg{F}}(\Phi(\eta_1), \Phi(\eta_2))$ is stable under the action of $G$. Since the Hom-sets are finite-dimensional vector spaces, it is clear that in this case they are generated linearly by irreducible tensors under $G$. So let $T_1, \dots T_n$ be some multiplet in $\Hom_{\alg{F}}(\Phi(\eta_1), \Phi(\eta_2))$ transforming according to the representation $\xi$. By the proof of Lemma~\ref{lem:irrdecompose} there is a $G$-invariant $X$ such that $T_i = X \Psi_i$, where the $\Psi_i \in H_\gamma$ form an orthonormal basis for $E(\gamma)$. Moreover, $\gamma$ is localised in some $\mc{O} \subset \spc{C}$ and transforms according to $\xi$.

Since $T_i \in \Hom_{\alg{F}}(\Phi(\eta_1), \Phi(\eta_2))$, we have, with $F = (A, \iota, 1) \in \alg{F}_0$,
\[
X \Psi_i \widehat{\eta_1}(\pi(F)) = \widehat{\eta_2}(\pi(F)) X \Psi_i = X \pi^{\spc{S}_a}(\gamma^{\spc{S}_a}(\eta_1(A) )) \Psi_i,
\]
where the last identity follows by applying Lemma~\ref{lem:vecmul} to the first term in the equation. Now, multiply on the right by $\Psi_i^*$, and sum over $i$. Since $\sum_{i=1}^d \Psi_i \Psi_i^* = I$ by~\cite[Proposition 270]{halvapp}, this leads to
\begin{equation}
	\label{eq:intertwiner}
	X \pi^{\spc{S}_a}(\gamma^{\spc{S}_a} \eta_1(A)) = \pi^{\spc{S}_a}(\eta_2(A)) X.
\end{equation}
By Lemma~\ref{lem:fieldfix}(ii) there is a $T \in \alg{A}^{\spc{S}_a}$ such that $\pi^{\spc{S}_a}(T) = X$, and by equation~\eqref{eq:intertwiner} and faithfulness of $\pi^{\spc{S}_a}$, we have $T \in \Hom_\alg{A}(\gamma \otimes \eta_1, \eta_2)$.
\end{proof}

\begin{corollary}
The tensor functor $\Phi$ is an embedding (i.e. faithful and injective on objects), but not full.
\end{corollary}
\begin{proof}
  It follows from Corollary~\ref{cor:restrict} that $\Phi$ is injective on objects. Since $\pi^{\spc{S}_a}$ is a faithful representation, Proposition~\ref{prop:extend} implies $\Phi$ is faithful. The preceding proposition implies that it is not full. Indeed, the image of $\Hom_{\alg{A}}(\eta_1, \eta_2)$ under the functor $\Phi$ is $\pi^{\spc{S}_a}(\Hom_{\alg{A}}(\eta_1, \eta_2))$, which in general is a proper subset of $\Hom_{\alg{F}}(\Phi(\eta_1), \Phi(\eta_2))$ as given by equation~\eqref{eq:homsets}.
\end{proof}

Inspired by the results of Doplicher and Roberts, M\"uger formulated a categorical version of the field net construction~\cite{MR1749250}. In a different context, a similar construction is due to Brugi\`eres~\cite{MR1741269}. In both approaches, modular categories are obtained by getting rid of a non-trivial centre. Here we investigate this in the present situation, c.f.~\cite{MR2183964}. We follow the approach of~\cite{MR1749250}, since it also works when the symmetric subcategory has infinitely many isomorphism classes of objects.

Let us recall the basic ideas in this construction. Suppose $\mc{C}$ is a braided tensor $C^*$-category and $\mc{S}$ is a full symmetric subcategory. By the Doplicher-Roberts theorem~\cite{MR1010160}, there is a unique compact group $G$ and an equivalence of categories $E: \mc{S} \to \operatorname{Rep}_f(G)$. In the case at hand, $\mc{C}$ is the category $\bfm^{\alg{A}}(\spc{C})$ and $\mc{S}$ is the symmetric subcategory $\dhr^{\alg{A}}(\spc{C})$.\footnote{Note that in the construction of the field net, the subcategory $\dhr^{\alg{A}}$ was used, without the localisation in $\spc{C}$. Using transportability, however, it is easy to see that one might as well choose $\dhr^{\alg{A}}(\spc{C})$, since this category is equivalent to $\dhr^{\alg{A}}$.} The group $G$ will be the symmetry group, and $E$ is the functor used in Section~\ref{sec:field}.

First a category $\mc{C} \rtimes_0 \mc{S}$ is defined. For each $k \in \widehat{G}$, choose a corresponding $\gamma_k \in \mc{S}$ such that $\mc{H}_k = E(\gamma_k)$ transforms according to $k$. The category $\mc{C} \rtimes_0 \mc{S}$ is the category with the same objects as $\mc{C}$, but with Hom-sets
\[
	\Hom_{\mc{C} \rtimes_0 \mc{S}}(\rho,\sigma) = \oplus_{k \in \widehat{G}} \Hom_{\mc{C}} (\gamma_k \otimes \rho, \sigma) \otimes \mc{H}_k,
\]
where the usual tensor product of vector spaces over $\mathbb{C}$ is used. One can then define a composition of arrows, a $*$-operation, conjugates, direct sums and in the case at hand, where the objects of $\mc{S}$ are degenerate, a braiding. Since the details are quite involved, we refer to the original paper~\cite{MR1749250}.

The category $\mc{C} \rtimes_0 \mc{S}$ already has most of the desired structure. One property, however, is missing: in general it is not closed under subobjects. To remedy this, a \emph{closure} construction is defined. This closure is denoted by $\mc{C} \rtimes \mc{S}$. It is called the \emph{crossed product} of $\mc{C}$ by $\mc{S}$. The basic idea is to add a corresponding (sub)object for each projection in $\Hom_{\mc{C} \rtimes_0 \mc{S}}(\eta,\eta)$. To make this precise: the category $\mc{C} \rtimes \mc{S}$ has pairs $(\eta, P)$ as objects where $\eta \in \spc{C}$ and $P = P^2 = P^* \in \Hom_{\mc{C} \rtimes_0 \mc{S}}(\eta,\eta)$. The morphisms are given by
\[
\Hom_{\mc{C} \rtimes \mc{S}}( (\eta_1,P_1),(\eta_2,P_2) ) = \{ T \in \Hom_{\mc{C} \rtimes_0 \mc{S}}(\eta_1, \eta_2)\, |\, T = T\circ P_1 = P_2 \circ T \},
\]
which is just $P_2 \circ \Hom_{\mc{C} \rtimes_0 \mc{S}}(\eta_1, \eta_2) \circ P_1$. Composition is as in $\mc{C} \rtimes_0 \mc{S}$. Because $P$ is a projection, $\id_{(\eta,P)} = P$. The tensor product can be defined by as $(\eta_1, P_1) \otimes (\eta_2, P_2) = (\eta_1 \otimes \eta_2, P_1 \otimes P_2)$, and the same as in $\mc{C} \rtimes_0 \mc{S}$ on morphisms. One can then show that $\mc{C} \rtimes \mc{S}$ is a braided tensor $C^*$-category with conjugates, direct sums and subobjects. The category $\mc{C}$ is embedded into the crossed product $\mc{C} \rtimes \mc{S}$ by a tensor functor $\iota: \mc{C} \to \mc{C} \rtimes \mc{S}$, defined by $\eta \mapsto (\eta, \id_\eta)$ and $\Hom_\mc{C}(\eta_1, \eta_2) \ni T \mapsto T \otimes \Omega$. Here $\Omega$ is a unit vector in the Hilbert space transforming according to the trivial representation of $G$. Like the functor $\Phi$, $\iota$ is a embedding functor that is not full.

The following proposition clarifies the relation between the crossed product $\bfm^{\alg{A}}(\spc{C}) \rtimes \dhr^{\alg{A}}(\spc{C})$ and the BF representations of the field net $\alg{F}$.
\begin{prop}
  The extension functor $\Phi:\bfm^{\alg{A}}(\spc{C}) \to \bfm^{\alg{F}}(\spc{C})$ factors through the canonical inclusion functor $\iota: \bfm^{\alg{A}}(\spc{C}) \to \bfm^{\alg{A}}(\spc{C}) \rtimes \dhr^{\alg{A}}(\spc{C})$. That is, there is a braided tensor functor $H: \bfm^{\alg{A}}(\spc{C}) \rtimes \dhr^{\alg{A}}(\spc{C}) \to \bfm^{\alg{F}}(\spc{C})$ such that the diagram
\[
\begin{diagram}
  \node{\bfm^{\alg{A}}(\spc{C})} \arrow{e,t}{\iota}\arrow{se,b}{\Phi} \node{\bfm^{\alg{A}}(\spc{C}) \rtimes \dhr^{\alg{A}}(\spc{C})} \arrow{s,r}{H} \\
  \node[2]{\bfm^{\alg{F}}(\spc{C})}
\end{diagram}
\]
commutes. Moreover, $H$ is full and faithful.
\end{prop}
\begin{proof}
  First define $H$ on the category $\bfm^{\alg{A}}(\spc{C}) \rtimes_0 \dhr^{\alg{A}}(\spc{C})$. Clearly, for objects $\eta$ we must set $H(\eta) = \Phi(\eta)$. In view of Proposition~\ref{prop:homsets}, it is natural to set for the morphisms $H(T\otimes\psi_k) = \pi^{\spc{S}_a}(T) \pi(I, \gamma_k, \psi_k)$, where $T \in \Hom_{\alg{A}}(\gamma_k \otimes \rho, \sigma)$, $\psi_k \in E(\gamma_k)$, and $k \in \widehat{G}$, and extend by linearity. It is not very difficult, although quite tedious, to verify that $H$ defines a strict braided monoidal functor from $\bfm^{\alg{A}}(\spc{C}) \rtimes_0 \dhr^{\alg{A}}(\spc{C})$ to $\bfm^{\alg{F}}(\spc{C})$. It is clear that $H$ is faithful, and by Proposition~\ref{prop:homsets} it is full.

  To define $H$ on the closure $\bfm^{\alg{A}}(\spc{C}) \rtimes \dhr^{\alg{A}}(\spc{C})$, consider one of its objects $(\eta, P)$.  By definition, $P^2 = P = P^* \in \Hom_{\bfm^{\alg{A}}(\spc{C}) \rtimes_0 \dhr^{\alg{A}}(\spc{C})}(\eta, \eta)$. It follows that $H(P)$ as defined above is a projection in $\Hom_{\alg{F}}(\Phi(\eta), \Phi(\eta))$. By localisation of $H(\eta)$ and Haag duality it follows that $H(P) \in \alg{F}(\spc{C})''$. Consider a spacelike cone $\widehat{\spc{C}}$ such that $\overline{\spc{C}} \subset \widehat{\spc{C}}$. Then by Property B there is an isometry $W \in \alg{F}(\widehat{\spc{C}}')'$ such that $WW^* = H(P)$. Now define $H(\eta, P)(\cdot) = W^* \widehat{\eta}(\cdot) W$. This defines a $*$-representation of $\alg{F}$ that is localised in $\widehat{\spc{C}}$, due to localisation properties of $E$. Using transportability, an equivalent representation localised in $\spc{C}$ can be obtained, in a similar way as done in Section~\ref{sec:bf}. Again it can be verified that $H$ is a braided monoidal functor. It is clearly faithful, and by Proposition~\ref{prop:homsets} and the definition of the Hom-sets in the crossed product, it is also full. Note that $H$ is not a \emph{strict} tensor functor, but only a strong one. This is due to the arbitrary choices one has to make in finding the isometry $W$, which is merely unique up to unitary equivalence.

  Finally, $\bfm^{\alg{A}}(\spc{C})$ is embedded in $\bfm^{\alg{A}}(\spc{C}) \rtimes \dhr^{\alg{A}}(\spc{C})$ by $\eta \mapsto (\eta, I)$. Hence $H \circ \iota(\eta) = H( (\eta, I) ) = \widehat{\eta}$, thus $H \circ \iota = \Phi$.
\end{proof}

\section{Essential surjectivity of $H$}
\label{sec:esssur}
One of our goals is to understand the category $\bfm^{\alg{F}}(\spc{C})$ in terms of the original AQFT $\mc{O} \mapsto \alg{A}(\mc{O})$. The functor $\Phi$ is not full, so it cannot provide a complete answer to this question. The functor $H$, however, \emph{is} full and faithful. Moreover, we have an explicit description of the crossed product in terms of our original net of observables $\alg{A}(\mc{O})$. Since a tensor functor is an equivalence of tensor categories if and only if it is an equivalence of categories~\cite{MR0338002}, it is enough to show that $H$ is an equivalence of categories. By the previous section $H$ is full and faithful, hence only essential surjectivity has to be shown. In this section this question is investigated. The first observation is that this is related to a property of the extension functor $\Phi$.
\begin{prop}
  \label{prop:esssur}
  The functor $H$ is essentially surjective if and only if $\Phi$ is \emph{dominant}. That is, for each irreducible $\eta \in \bfm^{\alg{F}}(\spc{C})$, $\eta \prec \Phi(\widetilde{\eta})$ for some $\widetilde{\eta} \in \bfm^{\alg{A}}(\spc{C})$.
\end{prop}
\begin{proof}
  Suppose first that $H$ is essentially surjective. Then for an irreducible object $\eta \in \bfm^{\alg{F}}(\spc{C})$, there is some $(\eta',P)$ such that $\eta \cong H(\eta',P)$. But by construction of $H$, evidently $H(\eta', P)$ is a subobject of $\Phi(\eta')$. Since $\eta \cong H(\eta',P)$, also $\eta \prec \Phi(\eta')$.

  Conversely, suppose $\Phi$ is dominant. Let $\eta \in \bfm^{\alg{F}}(\spc{C})$ be irreducible, and suppose $\eta'$ is such that $\eta \prec \Phi(\eta')$. Then there is a corresponding isometry $W \in \Hom_{\alg{F}}(\eta, \Phi(\eta'))$. Hence $W W^*$ is a projection in $\End_\alg{F}(\Phi(\eta'), \Phi(\eta'))$. Proposition~\ref{prop:homsets} shows that this projection comes from a corresponding projection $\widehat{P}$ in $\Hom_{\bfm^{\alg{A}}(\spc{C}) \rtimes_0 \dhr^{\alg{A}}(\spc{C})}(\eta', \eta')$, and we see that $\eta \cong H(\eta', \widehat{P})$. The result follows because $\bfm^{\alg{F}}(\spc{C})$ is semi-simple.
\end{proof}
In the remainder of this section, we comment on the question of finding conditions such that $\Phi$ is dominant. In the case of finite $G$ this problem has been solved in~\cite{MR1721563}. Given an irreducible sector of the field net, one can use the full $G$-spectrum of the field net to construct a direct sum that is $G$-invariant and contains $\eta$. This construction works in the present case of BF sectors as well. By Corollary~\ref{cor:restrict} it follows that this direct sum comes from extending a representation of the observable net.

A straightforward attempt to generalise this to arbitrary compact groups would be to replace the (finite) direct sum by a countable direct sum or even a direct integral. However, apart from convergence problems one might encounter, there is another issue: since the dimension $d(\eta)$ is strictly positive, and is additive under taking direct sums, this leads to a sector with infinite dimension. Hence it is not an element of our category $\bfm^{\alg{F}}(\spc{C})$.

Let us first recall how the group $G$ acts on the sectors, or more precisely, on equivalence classes of localised representations.
\begin{lemma}
  \label{lem:gaction}
  Let $\eta \in \bfm^{\alg{F}}(\spc{C})$. Then $G$ acts on equivalence classes $[\eta]$ by ${}^g[\eta] = [{}^g\eta] = [\alpha_g \circ \eta \circ \alpha_{g^{-1}}]$.
\end{lemma}
\begin{proof}
  This obviously defines an action. This action is well-defined: suppose $\eta_1(-) = V \eta_2(-) V^*$ for some unitary $V$. Then ${}^g\eta_2(-) = \alpha_g \circ \eta_2 \circ \alpha_{g^{-1}}(-) = \alpha_g(V \eta_1 \circ \alpha_{g^{-1}}(-) V^*) = \alpha_g(V) \alpha_g \circ \eta_1 \alpha_{g^{-1}}(-) \alpha_g(V^*)$, hence ${}^g\eta_1 \cong {}^g\eta_2$.
\end{proof}

The previous observations suggest that if there is any hope to construct a $G$-invariant direct sum of a sector of the field net, the action of $G$ on this sector should not be too ``wild'', in the sense that there should only be a finite number of mutually inequivalent sectors under the action of $G$. This is indeed a necessary condition, as will be shown below. This behaviour is described by the stabiliser subgroup.
\begin{definition}
	Suppose $\eta \in \bfm^{\alg{F}}(\spc{C})$. The \emph{stabiliser subgroup} $G_\eta$ is defined by $G_\eta = \{ g \in G\,|\, {}^g\eta \cong \eta \}$.
\end{definition}
By Lemma~\ref{lem:gaction} this is well-defined. Moreover, the index $[G:G_\eta]$ is finite if and only if there are only finitely many equivalence classes under the action of $G$. Note that $G_\eta$ is a closed subgroup of $G$, hence compact. The condition that the index be finite is necessary for finding a $G$-invariant dominating representation.
\begin{lemma}
	Suppose $\eta \prec \widehat{\eta}$ for $\eta \in \bfm^{\alg{F}}(\spc{C})$, where $\widehat{\eta}$ commutes with the action of $G$. Then $[G:G_\eta] < \infty$.
	\label{eq:finindex}
\end{lemma}
\begin{proof}
	Assume for simplicity that $\eta$ is irreducible; the general case readily follows. Decompose $\widehat{\eta} = \oplus_{i \in I} \eta_i$ where $I$ is some finite set. Then there is an $i \in I$ such that $\eta_i \cong \eta$, since $\eta \prec \widehat{\eta}$. Because ${}^g \widehat{\eta} = \widehat{\eta}$ for all $g \in G$, it follows that for every $g \in G$ there is some $j \in I$ such that ${}^g \eta_i \cong \eta_j$. As $g$ runs over $G$, $[{}^g \eta_i]$ runs over all equivalence classes ${}^g[\eta]$. It follows that there are at most $|I|$ such equivalence classes, or by the remark above: $[G:G_\eta] \leq |I|$.
\end{proof}

Our next goal is to construct a BF representation $\widehat{\eta}$ that commutes with the action of $G$, such that $\eta \prec \widehat{\eta}$. In other words: $\eta$ is a direct summand of $\widehat{\eta}$. Observe that it is enough to consider only summands $\eta_i \cong {}^{g_i}\eta$ for some $g_i \in G$. Now assume that $[G:G_\eta]$ is finite. Then there is a finite dimensional representation of $G$, permuting a basis of the space spanned by the left cosets $G\slash G_\eta$. Write $[g]$ for the coset of $g \in G$. Pick a representative $g_i$ of each coset. Since the field net has full $G$-spectrum, it is possible to find isometries $V_{[g_i]}$ such that $\alpha_g(V_{[g_i]}) = V_{[gg_i]}$ and the following relations hold:
\[
	V_{[g_i]}^* V_{[g_j]} = \delta_{i,j} I, \quad \sum_{[g_i] \in G\slash G_\eta} V_{[g_i]} V_{[g_i]}^* = I.
\]
Now if $g \in G$, there is a $g_j$ and a $h_j \in G_\eta$ such that $g g_i = g_j h_j$. Moreover, multiplication on the left induces a permutation on the cosets, hence also of the representatives $g_i$. Let $\widetilde{\eta}$ be such that $\eta \prec \widetilde{\eta}$. Consider $\widehat{\eta}(-) = \sum_{[g_i] \in G\slash G_\eta} V_{[g_i]} {}^{g_i}\widetilde{\eta}(-) V_{[g_i]}^*$. Then for $g \in G$,
\[
{}^g \widehat{\eta}(-) = \sum_{[g_i] \in G\slash G_\eta} \alpha_g(V_{[g_i]}) {}^{gg_i}\widetilde{\eta}(-) \alpha_g(V_{[g_i]})^* = \sum_{[g_i] \in G\slash G_\eta} V_{[g_i]} {}^{g_i}({}^{h_i}\widetilde{\eta}(-)) V_{[g_i]}^*,
\]
where $h_i$ is as above. So for $\widehat{\eta}$ to commute with the $G$-action, it is sufficient that ${}^h \widetilde{\eta} = \widetilde{\eta}$ for all $h \in G_\eta$. The existence of such a $\widetilde{\eta}$ is also necessary.

To find such an $\widetilde{\eta}$, by semi-simplicity of $\bfm^{\alg{F}}(\spc{C})$ it is enough to consider an irreducible $\eta$. We will do this in the rest of this section. By definition, for each $g \in G_\eta$ there is a unitary $v(g)$ such that ${}^g\eta(-) = v(g) \eta(-) v(g)^*$. By considering ${}^{gh}\eta = {}^g({}^h\eta)$ and using that $\eta$ is irreducible, it follows that
\[
	v(gh) = c(g,h) \alpha_g( v(h)) v(g), \quad g,h \in G_\eta,
\]
where $c(g,h)$ is a complex number of modulus one. In fact, it is not difficult to show that $c(g,h)$ is a 2-cocycle, with equivalence class $[c] \in H^2(G_\eta, \mathbb{T})$. The cohomology class does not depend on the specific choice of unitaries $v(g)$ and is the same for each $\eta' \cong \eta$. Hence $(G_\eta, [c])$ can be seen as an invariant of the sector. If $[c]$ is the trivial cohomology class, $v(g)$ is in fact an $\alpha$-one-cocycle and we can construct an $\eta' \cong \eta$ that commutes with the action of $G_\eta$, just as in the proof of Corollary~\ref{cor:restrict}.

The following observation, which amounts to the fact that the direct sum is independent of the chosen basis, turns out to be convenient.
\begin{lemma}
  Let $\eta \in \bfm^{\alg{F}}(\spc{C})$ be irreducible. Consider two direct sums of copies of $\eta$, $\widehat{\eta} = \sum_{i=1}^n V_i \eta(-) V_i^*$ and $\widehat{\eta}' = \sum_{i=1}^n W_i \eta(-) W_i^*$. Then $\widehat{\eta} = \widehat{\eta}'$ if and only if there is a unitary $n \times n$ matrix $\lambda$ such that $W_i = \sum_{i=1}^n \lambda_{ji} V_j$.
\end{lemma}
\begin{proof}
  ($\Rightarrow$) Define $\lambda_{ij}  = V_i^* W_j$, then $\lambda_{ij} \in \End_\alg{F}(\eta) \cong \mathbb{C}$, by irreducibility of $\eta$. By a straightforward calculation one easily verifies that $\lambda$ is indeed a unitary matrix, and $W_i = \sum_{i=n}^n \lambda_{ji} V_j$.

($\Leftarrow$) Easy calculation.
\end{proof}

Now suppose we have a direct sum $\widehat{\eta}(A) = \sum_{i=1}^n V_i \eta(A) V_i^*$. An easy calculation then shows that for $g \in G_\eta$:
\[
{}^g \widehat{\eta}(-) = \sum_{i=1}^n \alpha_g(V_i) v(g) \eta(-) v(g)^* \alpha_g(V_i^*),
\]
where the $v(g)$ are unitaries as above. Because $v(g)$ is unitary, it follows that $\alpha_g(V_i) v(g)$ is a basis of $\Hom_{\alg{F}}(\eta, {}^g\widehat{\eta})$. This space has a Hilbert space structure, defining an inner product by $\langle V,W \rangle I = W^*V$ for $V,W \in \Hom(\eta, {}^{g}\widehat{\eta})$. Combining this with the previous observations, we find the following necessary and sufficient criterion.
\begin{prop}
	\label{prop:ginv}
	There is a $G$-equivariant (i.e., commuting with the action of $G$) dominating sector $\widehat{\eta} \succ \eta$ if and only if the following conditions hold:
	\begin{enumerate}
		\item the stabiliser group $G_\eta$ has finite index in $G$, i.e. $[G:G_\eta] < \infty$,
		\item there is a finite-dimensional non-trivial Hilbert space $\mathcal{H}$ in $\alg{F}$ such that $\alpha_g(V) v(g) \in \mathcal{H}$ for all $V \in \mc{H}$ and $g \in G_\eta$.
	\end{enumerate}
\end{prop}
We end this section with a few remarks. First of all, the author unfortunately does not know of any physical interpretation of the conditions in the proposition. Furthermore it seems to be difficult to verify these conditions. However, the proposition generalises the situation where $G$ is finite. In this case, the conditions are trivially satisfied. If one can show that the cocycle $c(g,h)$ is trivial (as a cocycle in $H^2(G_\eta, \mathbb{T})$), it follows by Theorem~\ref{thm:cobound} that there is a unitary $w$ such that $v(g) = \alpha_g(w) w^*$. Condition (ii) is then satisfied by taking the one-dimensional Hilbert space spanned by $w$. Using Theorem~\ref{thm:cobound} one can show that $c(g,h)$ is trivial as a cocycle in the field net, which, however, is not sufficient here.

As a final remark, suppose that condition (ii) is satisfied. It follows that there is Hilbert space in $\alg{F}$ carrying a \emph{projective} unitary representation. Indeed, choose an orthonormal basis $V_i$ of $\mc{H}$. Then for $g \in G_\eta$, $\alpha_g(V_i) v(g)$ is a new basis for $\mc{H}$. Write $\lambda(g)$ for the unitary transformation that implements the basis change. It follows that $\lambda(gh) = c(g,h) \lambda(g) \lambda(h)$.

\section{Conclusions and open problems}
\label{sec:conclusions}
It would be desirable to arrive at a modular category starting from an AQFT in three dimensions, for example because of their relevance to topological quantum computing. In this paper some steps in this direction are taken. In particular, the category of stringlike localised or BF representations has many of the properties of a modular category. The existence of DHR sectors, which cannot be ruled out a priori, is shown to be an obstruction for modularity. To remove this obstruction, the original theory $\alg{A}$ is extended to the field net $\alg{F}$, which can be seen as a new AQFT without DHR sectors. The relation between those theories is partially made clear, in particular by the crossed product construction of Section~\ref{sec:crossed}. There is, however, one point that is not fully understood, namely the question whether the sectors in the new theory $\alg{F}$ can be completely described by the sectors of the theory $\alg{A}$. This is the case if for example $G$ is finite, or the conditions of Proposition~\ref{prop:ginv} hold for each BF sector of $\alg{F}$. In this case, the sectors of $\alg{F}$ are completely determined by the crossed product $\bfm^{\alg{A}} \rtimes \dhr^{\alg{A}}(\spc{C})$.

Although one major obstruction for modularity has now been removed, this is not enough to conclude that $\bfm^{\alg{F}}(\spc{C})$ is modular. In particular, there may be degenerate BF (but not DHR) sectors of $\alg{F}$. The other condition is that there should be only finitely many equivalence classes of BF representations of $\alg{F}$. In case the functor $H$ of Section~\ref{sec:esssur} is indeed an equivalence, both properties are determined by the crossed product, and hence ultimately by $\bfm^{\alg{A}}(\spc{C})$. In particular, in this situation, absence of degenerate sectors in $\bfm^{\alg{F}}(\spc{C})$ is equivalent to the absence of degenerate objects in  $\bfm^{\alg{A}}(\spc{C}) \rtimes \dhr^{\alg{A}}(\spc{C})$. This is essentially because $H$ is a \emph{braided} functor, which makes it possible to transfer the degeneracy condition of the braiding from one category to the other. The absence of degenerate objects of $\bfm^{\alg{A}}(\spc{C}) \rtimes \dhr^{\alg{A}}(\spc{C})$ is equivalent to the absence of degenerate \emph{BF} sectors (that are not DHR) of $\alg{A}$, since by~\cite{MR1749250} the crossed product has trivial centre if and only if $\dhr^{\alg{A}}(\spc{C})$ is equal to the centre of $\bfm^{\alg{A}}(\spc{C})$. The finiteness condition would follow by counting arguments from finiteness of $\bfm^{\alg{A}}(\spc{C})$. 

We give a list of some open problems and questions.
\begin{enumerate}
	\item In view of the remarks above, it would be interesting to understand the set of BF (that are not DHR) sectors of $\alg{A}$. In particular, are there conditions that imply that this set is finite, or does not contain any degenerate sectors? As for the latter: in the DHR case a condition for this was given in~\cite{MR1721563}. Perhaps this condition might be adapted to the case of BF sectors. It should be noted that both conditions (i.e. non-degeneracy and finiteness) are completely understood in the case of conformal field theory on the circle, in terms of an index of certain subfactors~\cite{MR1838752}. That method, however, cannot obviously be adapted to the case we are interested in, among other reasons because we have no condition for factoriality of the relevant algebras of observables. However, it would be interesting to know if there is an analogue of the condition of ``complete rationality'' that ensures modularity.
	\item It would be desirable to have a physical interpretation for the conditions given in Section~\ref{sec:esssur}. This might give some hints on how to prove these conditions in concrete theories.
	\item One of our assumptions was the absence of fermionic DHR sectors of $\alg{A}$. It would be interesting to see what can still be done if this assumption is dropped. In this case, the field net does not satisfy locality, but only \emph{twisted} locality. Thus one would lose the interpretation of $\alg{F}$ as an AQFT in the sense that it should only consist of observables commuting at spacelike distances.
	\item Can the techniques be useful in describing quantum spin systems? Such systems are more appropriate for topological quantum computing than relativistic quantum field theories, see e.g.~\cite{MR1951039}. There is some evidence that points into this direction~\cite{toricendo}. In particular, it can be shown that in Kitaev's $\mathbb{Z}_2$ model on the plane, single excitations can be described by automorphisms of an observable algebra. These automorphisms fulfill a selection criterion similar to the BF criterion. Moreover, they are localised and transportable, and using the methods here, one can explicitly calculate the statistics of these excitations. The results are consistent with Kitaev's results~\cite{MR1951039}. Although this simple model is by no means sufficient for quantum computing, it might be possible to extend the methods to more interesting models.
\end{enumerate}

\appendix
\section{}
In this appendix we collect some of the terminology regarding (tensor) categories and notions of superselection theory that will be used throughout the article. Due to lack of space, we restrict to the essentials. In particular, the categorical concepts can be defined much more generally than necessary for our purposes. For the essentials of category theory, details can be found in the book by Mac Lane~\cite{MR1712872}. For the structure of categories appearing in algebraic quantum field theory, see~\cite{mmappendix}. Modular categories are described in~\cite{MR1797619}. An overview of superselection theory can be found in the book of Haag~\cite{MR1405610}.

\subsection{Superselection theory}
A \emph{sector} is an unitary equivalence class of representations (satisfying some selection criterion such as the DHR or BF criterion) of the observable algebra. Representations satisfying the BF or DHR criterion can be described by localised and transportable endomorphisms of the observable algebra. Sometimes we will identify such an endomorphism $\rho$ with its sector, i.e., all unitary equivalent localised endomorphisms. These endomorphisms are the objects of a category, with intertwiners as morphisms. An intertwiner from $\eta_1$ to $\eta_2$ (and hence a morphism in $\Hom(\eta_1, \eta_2))$ is an operator $T$ such that $T \eta_1(A) = \eta_2(A) T$ for all observables $A$. There is a natural tensor product $\otimes$ (defined by composition of endomorphisms) on this category.

Another important concept is that of a \emph{conjugate sector}. A conjugate of a DHR or BF representation can be interpreted as an ``anti-charge''. Formally, a conjugate for a BF (or DHR) representation $\rho$ is a triple $(\overline{\rho}, R, \overline{R})$, where $\overline{\rho}$ is a BF (resp. DHR) representation. The operators $R, \overline{R}$ are intertwiners satisfying $R \in \Hom(\iota, \overline{\rho} \otimes \rho)$, with $\iota$ the trivial endomorphism, and $\overline{R} \in \Hom(\iota, \rho \otimes \overline{\rho})$ such that
\[
\overline{R}^*\rho(R) = I, \quad R^* \overline{\rho}(\overline{R}) = I, 
\] 
where $I$ is the unit of the observable algebra. If a conjugate exists, one can always choose a \emph{standard} conjugate. A conjugate $(\overline{\rho}, R, \overline{R})$ is called standard if $R^* \overline{\rho}(S) R = \overline{R}^* S \overline{R}$ for all $S \in \Hom(\rho,\rho)$. The conditions for a conjugate imply that $\overline{\rho} \otimes \rho$ (and $\rho \otimes \overline{\rho})$ contain a copy of the vacuum sector. A conjugate exists if and only if the sector has finite (statistics) dimension. The latter is then given by $d(\rho) I = R^*R$ with $R$ standard. Conjugates are intimately related to the statistics of a sector. It should be noted that conjugates can be defined in a much more general categorical setting, e.g.~\cite{MR1444286}.

In the category of BF representations, a braiding $\varepsilon_{\rho,\eta} \in \Hom(\rho \otimes \eta, \eta \otimes \rho)$ for every pair of objects $\rho, \eta$ can be defined. A sector is called \emph{degenerate} if, roughly speaking, it has trivial braiding with all objects. More precisely, $\rho$ is degenerate if and only if $\varepsilon_{\rho,\eta} \circ \varepsilon_{\eta, \rho} = I$ for all objects $\eta$. If this holds for a particular representative of a sector, it holds for all representatives of the sector. An object of the form $\iota \oplus \dots \oplus \iota$ is always degenerate. Degenerate sectors of this form are called \emph{trivial}.

\subsection{Category theory}
Let $F: \mc{C} \to \mc{D}$ be a functor. Then, for every pair of objects $\rho,\sigma$ in $\mc{C}$, there is a map $F_{\rho,\sigma} : \Hom(\rho,\sigma) \to \Hom(F(\rho), F(\sigma))$ defined by $S \mapsto F(S)$ for $S \in \Hom(\rho,\sigma)$. The functor $F$ is called \emph{faithful}, if $F_{\rho,\sigma}$ is injective for each pair of objects $\rho,\sigma$. Likewise, if it is surjective for all pairs, it is called \emph{full}. Note that a faithful functor is not necessarily injective on objects, that is, it might happen that $F(\rho) = F(\sigma)$ for distinct objects $\rho$ and $\sigma$ of $\mc{C}$. A faithful functor that is also injective on objects, is called an \emph{embedding}.\footnote{Note, however, that for some authors an embedding functor is only a faithful functor.} In particular, subcategories give rise to embedding functors. A subcategory of a category $\mc{C}$ is a category that contains a collection of the objects and morphisms of $\mc{C}$. A subcategory is called \emph{full} if it has the same morphisms as the bigger category, hence in that case it is completely determined by specifying its objects. Finally, a functor $F: \mc{C} \to \mc{D}$ is called an \emph{equivalence of categories} if it is full, faithful and essentially surjective, which means that for each object $D$ of $\mc{D}$, there is an object $C$ of $\mc{C}$ such that $F(C)$ is isomorphic to $D$. From a categorical perspective, equivalent categories are ``essentially the same''.

Certain categories admit a tensor (or monoidal) product $\otimes$. That is, one can form tensor products of objects and morphisms. In a tensor category there is a \emph{tensor unit} $\iota$, such that $\rho \cong \iota \otimes \rho \cong \rho \otimes \iota$, where $\cong$ means isomorphic in the category. Associativity is described by natural isomorphisms $\alpha_{\rho,\sigma,\tau} : \rho \otimes (\sigma \otimes \tau) \to  (\rho \otimes \sigma) \otimes \tau$ satisfying certain coherence conditions. A tensor category is called \emph{strict} if the associativity morphisms reduce to the identity, and $\rho \otimes \iota = \iota \otimes \rho = \rho$ for all objects $\rho$. The categories encountered in this paper are all strict. 

Every tensor category is monoidally equivalent to a strict tensor category. That is, there is a tensor functor between the two categories, that is also an equivalence of categories. A tensor functor is a functor $F$ together with natural isomorphisms $F(\rho \otimes \sigma) \to F(\rho) \otimes F(\sigma)$, and similarly for the tensor unit. Again, the functor is called strict if the isomorphisms are all identities. Even between \emph{strict} tensor categories, however, it might be necessary to consider non-strict tensor functors. In case both categories have a braiding, a \emph{braided} tensor functor $F: \mc{C} \to \mc{D}$ is a functor such that $F(\varepsilon^\mc{C}_{\rho,\sigma}) = \varepsilon_{F(\rho), F(\sigma)}^\mc{D}$ (or a suitably modified condition if the categories are not strict), where $\varepsilon^{\mc{C}}_{\rho,\sigma}$ is the braiding of $\mc{C}$.

The category of stringlike localised representations is called \emph{modular}, if it has only finitely many equivalence classes of representations and the centre (with respect to the braiding) is trivial. The latter condition is the statement that if $\varepsilon_{\rho, \eta} \circ \varepsilon_{\eta, \rho} = I$ for each object $\rho$, then $\eta = \iota \oplus \cdots \oplus \iota$, i.e., it is a direct sum of trivial endomorphisms. A modular category satisfies additional axioms (for example the existence of duals or conjugates), but these are automatically satisfied by the category of BF representations. The non-degeneracy condition is equivalent to Turaev's condition on a modular category~\cite{MR1292673}, which is stated in terms of invertibility of a certain matrix $S$, by a result of Rehren~\cite{MR1147467}.
\\
\\
{\bf Acknowledgements} This research is funded by NWO grant no. 613.000.608, which is gratefully acknowledged. I would also like to thank Michael M\"uger for valuable discussions and suggestions, and Klaas Landsman for a critical reading of the manuscript.
\bibliographystyle{abbrv}
\bibliography{refs}
\end{document}